\newif\ifsubmission
\newif\ifpreprint

\submissionfalse
\preprinttrue

\documentclass[12pt]{article}
\usepackage[utf8]{inputenc}
\usepackage{natbib}
\usepackage{parskip}
\usepackage{bibentry}
\usepackage{booktabs}
\usepackage{csquotes}
\usepackage{comment}
    
\usepackage{caption, subcaption}
\DeclareCaptionLabelFormat{no-parens}{\textsf{\textbf{#2}}}
\usepackage[linkcolor=black,colorlinks=true,citecolor=black,filecolor=black]{hyperref}
    
\usepackage{enumitem}
    
\usepackage{comment}
\usepackage{amsthm}
\usepackage{thm-restate, amsmath, amssymb, amsfonts}
\usepackage{mathtools}
\usepackage{algorithm}
\usepackage{algpseudocode}
\declaretheorem[name=Theorem]{thm}

\ifpreprint
    \usepackage[margin=1.2in]{geometry}
\fi

\ifsubmission
    \usepackage[margin=1.375in]{geometry}
\fi

\begin{document}

\title{Teaching signal synchronization in deep neural networks with prospective neurons}

\ifsubmission
    \author{
      Nicolas Zucchet$^{1, \dagger}$ \quad
      Qianqian Feng$^{2, \dagger}$ \quad
      Axel Laborieux$^{3}$ \\
      Friedemann Zenke$^{3,4}$ \quad
      Walter Senn$^{5,\ddagger}$ \quad
      João Sacramento$^{6,\ddagger}$ \\[1.5em]
      \small $^1$Computer Science Department, ETH Zürich \\
      \small $^2$Gatsby Computational Neuroscience Unit, University College London\\
      \small $^3$Friedrich Miescher Institute for Biomedical Research \\
      \small $^4$Faculty of Science, University of Basel \\
      \small $^5$Department of Physiology, University of Bern \\
      \small $^6$Google, Paradigms of Intelligence Team \\[1em]
      \small $^\dagger$These authors contributed equally to this work \\
      \small $^\ddagger$These authors contributed equally to this work \\[1em]
    }
\fi
\ifpreprint
    \author{
      Nicolas Zucchet$^{1,}$\thanks{\vspace{0.5cm}Shared first authorship. $^\dagger$ Shared senior authorship. Correspondance to \texttt{nzucchet@ethz.ch}.} \quad
      Qianqian Feng$^{2, \,*}$ \quad
      Axel Laborieux$^{3}$ \\
      Friedemann Zenke$^{3,4}$ \quad
      Walter Senn$^{5,\dagger}$ \quad
      João Sacramento$^{6,\dagger}$ \\[1.5em]
      \small $^1$Computer Science Department, ETH Zürich \\
      \small $^2$Gatsby Computational Neuroscience Unit, University College London\\
      \small $^3$Friedrich Miescher Institute for Biomedical Research \\
      \small $^4$Faculty of Science, University of Basel \\
      \small $^5$Department of Physiology, University of Bern \\
      \small $^6$Google, Paradigms of Intelligence Team
      \vspace{-0.3cm}
    }
\fi

\date{}

\maketitle

\begin{abstract}
\noindent
Working memory requires the brain to maintain information from the recent past to guide ongoing behavior. Neurons can contribute to this capacity by slowly integrating their inputs over time, creating persistent activity that outlasts the original stimulus. However, when these slowly integrating neurons are organized hierarchically, they introduce cumulative delays that create a fundamental challenge for learning: teaching signals that indicate whether behavior was correct or incorrect arrive out-of-sync with the neural activity they are meant to instruct. Here, we demonstrate that neurons enhanced with an adaptive current can compensate for these delays by responding to external stimuli prospectively -- effectively predicting future inputs to synchronize with them. First, we show that such prospective neurons enable teaching signal synchronization across a range of learning algorithms that propagate error signals through hierarchical networks. Second, we demonstrate that this successfully guides learning in slowly integrating neurons, enabling the formation and retrieval of memories over extended timescales. We support our findings with a mathematical analysis of the prospective coding mechanism and learning experiments on motor control tasks. Together, our results reveal how neural adaptation could solve a critical timing problem and enable efficient learning in dynamic environments.
\end{abstract}

\ifpreprint
    \break
\fi 

Understanding how the brain learns and adapts to new situations is a central question in neuroscience. Effective learning requires associating current neural activity with the delayed consequences of behavior, whether from inherent factors such as delayed rewards or processing delays within the brain. Such processing delays are accentuated by the brain's hierarchical organization, where different areas specialize in various aspects of information processing \citep{kandel_principles_2000}. Indeed, working memory requires neurons to slowly integrate their inputs over time to maintain information from the recent past, but these slow dynamics introduce processing delays that accumulate across hierarchical levels. This creates a fundamental problem for learning: when teaching signals -- whether external rewards or internal error signals -- arrive at output-related areas, the neural activity in earlier layers is already out of sync due to cumulative processing delays, creating a temporal mismatch between the activity that should be instructed and the teaching signals guiding learning.

Error backpropagation \citep{werbos_beyond_1974, rumelhart_learning_1986} has emerged as the canonical algorithm for solving the spatial credit assignment problem in deep artificial neural networks, with numerous recent propositions exploring its potential biological implementations \cite[e.g.,][]{lee_difference_2015, lillicrap_random_2016, scellier_equilibrium_2017, whittington_approximation_2017, richards_deep_2019, meulemans_least-control_2022}. However, these models typically assume instantaneous transmission, failing to account for the slow dynamics of the neurons responsible for working memory. To address this limitation, researchers have proposed solutions including forward models \citep{miall_forward_1996, wolpert_internal_1998, gilra_predicting_2017}, eligibility traces \citep{izhikevich_solving_2007, gerstner_eligibility_2018, zenke_superspike_2018, bellec_solution_2020}, and prospective neurons \citep{mi_spike_2014, haider_latent_2021, senn_neuronal_2024}. In this work, we focus on the latter, examining how prospective neurons can mitigate delays, both theoretically and practically, and investigating how adaptive neurons can serve as practical implementations of such prospective mechanisms.

We structure the paper as follows. First, we focus on the theoretical foundations of prospective neurons by formalizing delay as a deviation from an ideal trajectory that needs to be tracked -- specifically, what neural activity would look like if no processing delays existed. Through this lens, we demonstrate that prospective neurons can effectively circumvent delays after an initial warm-up phase, whereas standard slowly integrating neurons prove insufficient. Next, we explore adaptive neurons as physically realistic implementations of such prospective neurons, analyzing how this physical implementation impacts both tracking and learning performance. Finally, we validate our theoretical insights on a range of behaviorally relevant motor control learning problems, including scenarios with delayed rewards that require tuning the memory storage process. Our findings reveal that the tracking ability of prospective neurons enables effective teaching signal synchronization, thus facilitating online learning in neural networks and providing insights into how biological systems might overcome internal signal propagation delays.
\section{Theoretical properties of prospective neurons}
\label{sec:theory}

Biological neural networks comprise dynamical neurons that integrate information over time, with leaky integrators \citep{abbott_lapicques_1999} standing as the archetypal model for such neurons. The temporal integration of neural signals presents a critical challenge, particularly for neurons responsible for propagating errors, as it results in a misalignment between neural activities and their corresponding learning signals. In this section, we first formalize such delays as deviations from a reference trajectory obtained under an idealized, instantaneous system. We then analyze how the network structure, neuronal characteristics, and external stimuli properties interplay to influence these deviations. Finally, we investigate prospective dynamics as a potential solution to mitigate temporal delays.

\subsection{Teaching signal synchronization as a tracking problem}
\label{subsec:synchro_as_tracking}

While our ideas apply to a broad class of plasticity rules, we focus our exposition on gradient-based learning and the backpropagation algorithm, as it is the canonical algorithm for learning deep neural networks. In this context, teaching signals are error signals derived from an objective function.

Let us consider a feedforward neural network whose neurons' voltages $u^l$ at layer $l$ are determined by
\begin{equation}
    \label{eqn:BP-forward}
    u^{l+1} = W^{l+1} \rho(u^l) ~~ \mathrm{and} ~ u^0 = x
\end{equation}
where $\rho$ is a nonlinear activation function, $W^{l+1}$ the synaptic strength matrix, and $x$ the input of the network, which we consider fixed for now. The voltage of the last layer, $u^L$, is taken to be the output of the network and is compared with a target $y$, through the loss function $\ell (u^L, y)$. Error signals $\delta$ are backpropagated through the network hierarchy with
\begin{equation}
    \label{eqn:BP-backward}
    \delta^{l} = \rho'(u^l) {W^{l+1}}^\top \delta^{l+1} ~~ \mathrm{and} ~ \delta^L = \nabla  \ell(u^L, y).
\end{equation}
Those equations follow from defining the gradient of the loss with respect to the postsynaptic voltages, $\delta^l := \nabla_{u^l} \,\ell$, and recursively applying the chain rule to this quantity. When using the mean squared error loss, the error signal for the output units becomes $\delta^L = u^L - y$.
Error signals can then be combined with voltages to obtain the gradient following update $\Delta W^{l+1} \propto -\nabla_{W^l} \,\ell(u^L, y) = -\delta^{l+1} {\rho(u^l)}^\top$.
While these equations provide a theoretically grounded learning rule, they ignore temporal dynamics: voltages and errors are assumed to propagate instantaneously through the network hierarchy.
It can be justified by a separation of timescales argument: the input $x$ and target output $y$ evolve at a slower rate than the neurons, so that they can be considered constant. Through the instantaneous nature of signal processing in this model, error signals always arrive on time.

However, realistic models of learning should take into account scenarios in which external stimuli evolve faster than the response time of the considered neural network. The separation of timescales argument mentioned above thus no longer holds. We must take into account the temporal dynamics of neural activity. Here, we study a standard leaky integrator model, in which $(u_t, \delta_t)$ satisfies the differential equation
\begin{equation}
    \label{eqn:continuous-time-BP}
    \begin{split}
        \tau \dot{u}_t &= -u_t + W\rho(u_t) ~~\mathrm{and}~~ u_t^0 = x_t\\
        \tau \dot{\delta}_t &= -\delta_t + \rho'(u_t)  W^\top \delta_t  ~~\mathrm{and}~~ \delta^L_t = \nabla \ell(u^L_t, y_t)
    \end{split}
\end{equation}
where $\dot{u}$ (resp. $\dot{\delta}$) stands for the temporal derivative of $u$ (resp. $\delta$), and $\tau$ is the neuronal integration time constant, which we consider to be the same for all $u$ and $\delta$ for simplicity. For conciseness, we have concatenated the voltages $u^l$ (resp. $\delta^l$) of all layers into a single vector $u$ (resp. $\delta$). The matrix $W$ contains all the $W^l$ matrices in its lower block diagonal. Equilibrium states of the dynamics (\ref{eqn:continuous-time-BP}) satisfy the backpropagation equations (\ref{eqn:BP-forward}) and (\ref{eqn:BP-backward}), but they are never reached when $(x_t, y_t)$ evolves rapidly over time. As a corollary, error signals always arrive with some delay, which increases with the depth of the network and the characteristic time constants of the neurons involved.

We formalize the intuition presented above by considering the trajectory of neural activities obtained by instantaneously processing external stimuli as an ideal target that the true neural activity should follow. Stated as such, neurons must solve a tracking problem, and tracking error becomes a clear metric to quantify delay. In the perfect tracking limit, the network behaves like its instantaneous counterpart. Formally, let $s_t$ be a vector summarizing the state of the system variables at time $t$ (the voltages $u_t$ and error signals $\delta_t$ in our previous example). The input $s$ receives information from both external stimuli and the rest of the network is $f_\theta(s_t, t)$. It depends on the parameters $\theta$ of the network (e.g., the synaptic weights $W$), as well as on $x_t$ and $y_t$ through $t$. We define a \textit{target trajectory} $s^\ast$ through the self-consistency equation
\begin{equation}
    \label{eqn:equilibrium-tracking}
    s_t^\ast = f_\theta(s_t^\ast, t)
\end{equation}
for all $t$. The instantaneous backpropagation equations (\ref{eqn:BP-forward}) and (\ref{eqn:BP-backward}), as well as the ones underlying a large variety of more biologically plausible learning algorithms, fit under this framework, as we shall see in Section~\ref{subsec:learning-teacher-student}. It is important to note that the target trajectory $s^*$ is implicitly defined, in contrast to more traditional control problems where targets are directly provided to the system. Additionally, when $f_\theta$ encodes cyclic dependencies through, e.g., recurrent connections, multiple such trajectories can exist.

\subsection{Leaky integrators cannot track target trajectories} 
\label{sub-sec:flow}

\begin{figure}[t]
    \centering
    \includegraphics[scale=1.2]{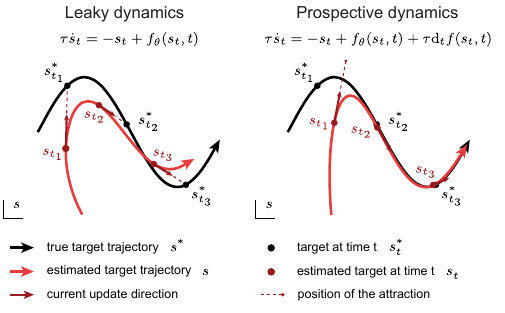}
    \caption{\textbf{Prospective dynamics can track the target trajectory whereas leaky dynamics cannot.} Under the leaky integrator dynamics, the current state $s_t$ is attracted towards the current target $s_t^\ast$. By the time it arrives there, the target will have already moved away, so $s$ is always lagging behind $s^\ast$. Adding a prospective input to the leaky dynamics enables perfect tracking, as the position of $s^\ast$ in the near future is indirectly estimated and will serve as a target for the $s$ dynamics.}
    \label{fig:intuition}
\end{figure}

We demonstrate that leaky integrators satisfying the differential equation $\tau \dot{s}_t = -s_t + f_\theta(s_t, t)$ cannot track the reference trajectory without a strict separation of timescales. Specifically, Theorem~\ref{thm:problem_flow} states that, under mild regularity assumptions on the structure of $f_\theta$, the target trajectory $s_t^\ast$ can only be approached up to some error. This error decreases when the network gets faster or external inputs, such as $x_t$, slow down.
\begin{thm}
    \label{thm:problem_flow}
    Let $f_\theta$ be such that the largest eigenvalue of the symmetric part of the Jacobian $\partial_s f_\theta(s, t)$ is always smaller than a constant $1-\mu$ for $\mu>0$. Let $s^\ast_t$ be a target trajectory that satisfies $s^\ast_t = f_\theta(s^\ast_t, t)$ for all $t$ and $\lVert \dot{s}^\ast_t \rVert \leq \gamma$. Then, the trajectory of states $s_t$ obtained by integrating the leaky dynamics $\tau \dot{s}_t = -s_t + f_\theta(s_t, t)$ satisfies
    \begin{equation*}
        \underset{t\rightarrow \infty}{\mathrm{limsup}} ~ \lVert s_t - s^\ast_t \rVert \leq \frac{\gamma\tau}{\mu}.
    \end{equation*}
    Furthermore, this bound is tight; that is, we can find a $f$ for which the upper bound is reached.
\end{thm}
This result formalizes the intuition that the trajectory $s$ obtained through leaky integration lags behind the target $s^*$. More precisely, $s$ asymptotically lies in a tube centered on $s^\ast$. Intuitively, this occurs as $s$ is attracted to the current target, which will have already moved away when reached, cf. Figure~\ref{fig:intuition}. The ratio $\tau/\mu$ reflects the effective time constant of the system: $\tau$ corresponds to the time constant of individual neurons, while $1/\mu$ is linked to network synaptic weights. For instance, as the depth of the network increases, so does $1/\mu$. Under constant external inputs, this constant also represents the worst exponential convergence rate that the system can exhibit. The $\gamma$ constant reflects both network geometry and the rate of change in external inputs. We provide more details on these considerations and a proof of Theorem~\ref{thm:problem_flow} in Appendix~\ref{app:theoretical_derivations}.

This result provides a tight upper bound for the tracking error $\lVert s_t - s_t^\ast \rVert$. We now verify whether the characterization it provides holds in practice. To this end, we consider a fully connected feedforward neural network with two hidden layers, whose neural activity and corresponding error signals follow the dynamics of (\ref{eqn:continuous-time-BP}). We provide random linear combinations of sine waves as input to the network and compare its outputs with those of an instantaneous feedforward teacher network receiving the same input. As the quantities involved in the theorem can be costly to estimate directly, we use proxies in practice. We consider the maximal value of $\lVert-s_t + f_\theta(s_t, t) \rVert$ after an initial convergence phase as a proxy for the tracking error of Theorem~\ref{thm:problem_flow}. This is a reasonable choice, as these quantities are linearly correlated under the theorem's assumptions. We refer to this metric as the \textit{empirical tracking error} and use it to measure how out-of-sync error signals arrive. 

We vary two parameters: the time constant $\tau$ of the neurons and the typical angular velocity $\omega_0$ at which the sinusoidal inputs vary (used as a proxy for $\gamma$; see Methods). All other factors, including neural network weights, remain fixed. Figure~\ref{fig:theory}A illustrates how the empirical tracking error evolves as a function of $\tau$ and $\omega_0$. Our empirical results confirm the theory: better tracking is only reliably achievable through timescale separation, that is, as $\tau \omega_0$ approaches 0.

\begin{figure}
    \centering
    \includegraphics[scale=1.2]{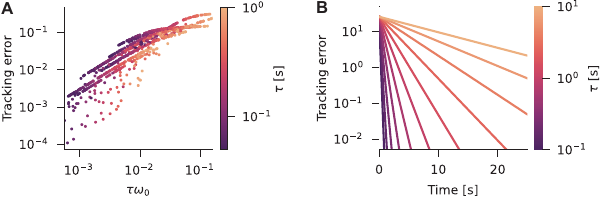}
\caption{\textbf{Theoretical predictions on tracking properties of leaky and prospective neurons hold in simulation}. \textbf{A}. Empirical tracking error for the leaky dynamics (\ref{eqn:continuous-time-BP}) as a function of the neurons time constant $\tau$ times the characteristic angular velocity of the inputs $\omega_0$. This quantity is directly linked to the theoretical bound of Theorem~\ref{thm:problem_flow}. Each dot corresponds to random $\tau$ and $\omega_0$ values. As predicted by Theorem~\ref{thm:problem_flow}, the tracking error scales linearly with $\tau\omega_0$ in the worst case, which measures how slow the neurons are compared to the input. Consequently, teaching signals arrive out-of-sync. \textbf{B}. Under the prospective dynamics, the instantaneous empirical tracking error converges exponentially fast to 0, with the neurons time constant $\tau$ modulating the convergence speed, consistently with Theorem~\ref{thm:perfect_tracking}.}\label{fig:theory}
\end{figure}

\subsection{Prospective neurons can track target trajectories}
\label{sub-sec:prospective}

As argued in the previous section, vanilla leaky integrator dynamics cannot achieve such tracking as the state $s_t$ is attracted to the current target and is thus always lagging behind the target $s^\ast_t$. Instead, if we manage to estimate what the target state will be in the future and use it as an implicit target for neural activity, the gap between $s_t$ and $s_t^\ast$ will progressively vanish, as illustrated in Figure~\ref{fig:intuition}. In the following, we show that slightly modifying the leaky dynamics by adding a prospective input to the inputs that leaky neurons receive will asymptotically (on a time scale of $\tau$) lead to perfect target tracking. As a consequence, prospective neurons behave like instantaneous neurons when given enough adaptation time.

More concretely, we introduce prospective dynamics that estimate what the input $f_\theta(s_t, t)$ each neuron receives will be in $\tau$ seconds through a first-order approximation and add it to the input. That is, we consider the state dynamics
\begin{equation}
    \label{eqn:prospective_dynamics}
    \tau \dot{s}_t = -s_t + f_\theta(s_t, t) + \tau \mathrm{d}_t f_\theta(s_t, t).
\end{equation}
This extra input enables leaky neurons to perfectly track equilibrium, as demonstrated in Theorem~\ref{thm:perfect_tracking}; see Appendix~\ref{app:theoretical_derivations} for a proof. Figure~\ref{fig:theory}B shows that the Theorem holds in practice, using the same experimental setup as in the previous section.
\begin{thm}
    \label{thm:perfect_tracking}
    Let $s$ follow the prospective dynamics (\ref{eqn:prospective_dynamics}). Then, assuming that $\partial_s f_\theta(s_t,t)$ is always invertible and $f_\theta$ is Lipschitz continuous in $t$ on that trajectory, we have 
    \begin{equation*}
        \lVert s_t - f_\theta(s_t, t) \rVert = \lVert s_0 - f_\theta(s_0, 0) \rVert \exp \left (-\frac{t}{\tau} \right )
    \end{equation*}
    and
    \begin{equation*}
        \underset{t \rightarrow \infty}{\mathrm{limsup}} ~ \lVert s_t - s_t^\ast \rVert = 0,
    \end{equation*}
    i.e., after an initial exponential convergence phase, $s_t$ and $s_t^*$ coincide.
\end{thm}

We conclude this section by noting that prospective dynamics precede our work. Most closely related are \cite{mi_spike_2014}, \cite{haider_latent_2021}, and \cite{senn_neuronal_2024}. The discussion section details connections with existing work from theoretical neuroscience as well as from other fields.

\section{Bio-physical implementation of prospective neurons}
\label{sec:physical_constraints}

So far, we have shown that prospective inputs enable leaky neurons to solve the tracking problem inherent to teaching signal synchronization, but we have ignored all details regarding their physical implementation.
We now delve into two of these crucial details: We first show how a mechanism appearing in adaptive neuron models \citep{brette_adaptive_2005, gerstner_neuronal_2014} can implement prospective dynamics plausibly by subtracting from the instantaneous input current a low-pass filtered version of itself, provided by an adaptation current.
Then, we study the robustness of the proposed dynamics to mismatches between the neurons' and prospective inputs' time constants, since we cannot expect them to perfectly match in biological neurons.

\subsection{Adaptive neurons can be prospective}

Prospective dynamics are a theoretical ideal that enables the efficient tracking of the instantaneous neural activity trajectory and thus remove delays, but they require the physical estimation of time derivatives. In analog filter design \citep{agarwal_foundations_2005}, it is well established that true differentiators are not physically realizable, as they require an infinite amount of energy. One simple, realizable differentiator is the high-pass filter. More precisely, if $a$ is a low-pass filter of $f_\theta(u_t, t)$ following
\begin{equation}
    \label{eqn:adaptivea}
    \tau_a \dot{a}_t = -a_t + f_\theta(u_t, t),
\end{equation}
then $\tau_a^{-1} (f_\theta(u_t, t) - a_t)$ converges to $\mathrm{d}_t f_\theta(u_t, t)$ when $\tau_a$ converges to 0. While it introduces some approximation in the estimation of the $\mathrm{d}_t f_\theta(u_t, t)$, it is more robust to noise. Plugging this estimate into the prospective dynamics (\ref{eqn:prospective_dynamics}) gives
\begin{equation}
    \label{eqn:adaptive}
    \begin{split}
        \tau \dot{u}_t &= -u_t + \left ( 1 + \frac{\tau}{\tau_a} \right )f_\theta(u_t, t)  - \frac{\tau}{\tau_a} a_t
    \end{split}
\end{equation}
Such neural dynamics are reminiscent of adaptive neuron models, in which $a$ can be interpreted as an adaptive current. There exist subtle differences with the adaptive neuron model: the adaptive current here integrates the input current $f_\theta$ and not the voltage, the leak term is linear, and the input and adaptive currents have specific weighting factors. We discuss potential physiological implementations of this mechanism in the discussion.

Next, we wonder how effective such adaptive neurons are at tracking the target trajectory and, in particular, whether the adaptive current improves the tracking ability of leaky neurons. To that extent, we study the linearization of the adaptive dynamics around $\tau_a = 0$ up to the first order in $\tau_a$: we show in Appendix~\ref{app:physical_implementation} that
\begin{equation}
    \tau \dot{u}_t = -u_t + f_\theta(u_t, t) + \tau \mathrm{d}_t f_\theta(u_t, t) + \tau \tau_a \mathrm{d}_t^2 f_\theta(u_t, t) + O(\tau \tau_a^2)
\end{equation}
Intuitively, up to a first-order approximation, the tracking error of adaptive dynamics scales with $\tau_a \tau$ when it scales with $\tau$ for the vanilla leaky dynamics (Theorem \ref{thm:problem_flow}). For small enough $\tau_a$, the former are thus better trackers than the latter. Additionally, biological neurons may be implementing more elaborate differentiating circuits, which would make the gap even larger.

We now measure the impact of adding an adaptive current and that of the integration time on the tracking error effect, as shown in Figure~\ref{fig:teacher_student}C, using the teacher-student setting from the previous section. We vary the neuronal time constant $\tau$ and consider different $\tau_a/\tau$ ratios. We make two observations: First, for all $\tau$ values considered, the adaptive dynamics are better trackers than the leaky dynamics, even when $\tau_a$ is greater than $\tau$. Second, the tracking error also approaches 0 more quickly as $\tau$ converges to 0.
Those results confirm that the addition of adaptive currents to leaky neurons improves their tracking abilities.

\begin{figure*}[t]
    \centering
    \includegraphics[scale=1.2]{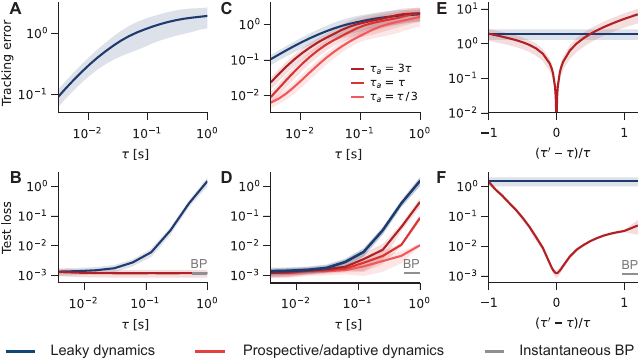}
    \captionof{figure}{\textbf{Better target trajectory tracking translates to greater learning performance.} Top row measures the empirical tracking error in the setup of Sections~\ref{sec:theory} and \ref{sec:physical_constraints}, bottom row the test loss in the teacher-student task of Section~\ref{subsec:learning-teacher-student}. Note that the two setups are similar but have some differences, cf. Methods. \textbf{A, B.} Leaky dynamics (\ref{eqn:continuous-time-BP}) require $\tau \rightarrow 0$ to perfectly track the reference trajectory and thus effectively learn (blue). Prospective dynamics (\ref{eqn:prospective_dynamics}) can do so for any $\tau$ (red), matching the performance of instantaneous backpropagation (grey). \textbf{C, D} Adaptive dynamics (\ref{eqn:adaptive}) track and learn better than leaky dynamics and the effect amplifies as $\tau_a$ gets smaller, relatively to $\tau$. \textbf{E, F.} Time constant mismatches as in (\ref{eqn:mismatch_tau}) hinder tracking and thus learning in the prospective dynamics. Yet, for a wide range of mismatches, the learning loss remains significantly lower that the one of leaky dynamics.}
    \label{fig:teacher_student}
\end{figure*}

\subsection{Time constant mismatches' impact on prospectiveness is minimal}
\label{subsec:time_constant_mismatch}

Besides the requirement for a differentiator circuit, which we discussed in the previous section, our theoretical analysis also relies on the assumption that all neurons have the same time constant $\tau$ and that it matches that of their prospective inputs. The first one is straightforward to relax, as our theory equally applies to this regime. In the following, we study the impact of relaxing the second assumption.

Let us consider a simple prospective network in which
\begin{equation}
    \tau \dot{u}_t = -u_t + W\rho(u_t) + \tau \mathrm{d}_t[W\rho(u_t)].
\end{equation}
We ignore teaching signals here for the sake of simplicity, but the following analysis naturally extends to that case.
In the previous equation, the prospective input can be computed pre-synaptically by leveraging the fact that $W \rho(u_t) + \mathrm{d}_t [ W \rho(u_t) ] = W ( \rho(u_t) + \tau \mathrm{d}_t \rho(u_t) )$, or post-synaptically, by receiving the input $W\rho(u_t)$ and then computing its prospective value.
We argue that the former is more sensitive to time constant mismatches than the latter. Indeed, if each neuron has a different prospective time constant, the different firing rates received by a neuron will have different underlying time constants. 
However, only two time constants are involved when prospective inputs are computed post-synaptically. 
We thus consider that prospective inputs are computed post-synaptically in the following manner and note that this is consistent with adaptive input currents.

Let us now analyze the impact of a time constant mismatch on the tracking properties of the prospective dynamics. Such a mismatch can occur in the adaptive dynamics if the constants in front of $f_\theta$ and $a_t$ do not match those prescribed by our theory. 
To keep the analysis simple, we assume that all neurons have a time constant equal to $\tau$ and that the time constant within the prospective input is $\tau'$. The network dynamics introduced above becomes
\begin{equation}
    \label{eqn:mismatch_tau}
    \tau \dot{u}_t = -u_t + W\rho(u_t) + \tau' \mathrm{d}_t [W \rho(u_t)].
\end{equation}
The deviation from the target trajectory can be expressed through a first-order approximation as
\begin{equation}
    \label{eqn:mismatch_taylor}
    u_t = u_t^\ast + (\tau' - \tau) {\rm d}_\tau u_t + O((\tau' - \tau)^2)
\end{equation}
and we are interested in finding an expression for ${\rm d}_\tau u_t$ to quantify the effect of the mismatch.
By injecting Eq.~\ref{eqn:mismatch_taylor} into Eq.~\ref{eqn:mismatch_tau}, if $\rho$ is close to being linear, we show in Appendix~\ref{app:physical_implementation} that the deviation ${\rm d}_\tau u_t$ satisfies the differential equation
\begin{equation}
    \label{eqn:mismatch_diffeq}
    \tau \dot{({\rm d}_\tau u_t)} = - {\rm d}_\tau u_t + \dot{u}_t^\ast. 
\end{equation}
It follows that the first-order deviation from the target in $\tau - \tau'$ grows as the target moves faster and that $\lVert u_t - u_t^\ast \rVert \leq |\tau - \tau'| \gamma + O\left((\tau-\tau')^2\right)$ at all time $t$, with $\gamma$ as in Theorem~\ref{thm:problem_flow}.
A consequence of this result is that faster inputs increase the sensitivity of the tracking to time constant mismatches, as they lead to faster $u^\ast$ and thus larger $\gamma$. We also note that this is a local result around $\tau'=\tau$. Larger deviations can qualitatively change the picture. For example, the dynamics can be contractive for $\tau' = \tau$, but explosive for some $\tau' > \tau$, meaning that the distance between $u_t^*$ and $u_t$ can no longer be uniformly bounded over time. Appendix~\ref{app:physical_implementation} provides a detailed example.

We simulate a time constant mismatch in Figure~\ref{fig:teacher_student}E and observe that prospective dynamics lead to better tracking abilities than leaky dynamics for a wide range of prospective time constants $\tau'$. Our theoretical analysis suggests that as $\tau$ gets smaller, the tracking error grows more slowly as a function of $(\tau'-\tau)/\tau$, so the negative peak we observe around $0$ becomes flatter. Fast neurons will thus be relatively less sensitive to a time constant mismatch. Additionally, the prospective and leaky dynamics match for $\tau'=0$, so we can reasonably expect prospective dynamics to exhibit better tracking for all $\tau' < \tau$ in general, similarly to what we observe in Figure~\ref{fig:teacher_student}.
\section{Learning with prospective neurons}
\label{sec:experiments}

Previous sections established the appealing theoretical properties of prospective neurons in terms of tracking, as well as their robustness to physical implementations. In this section, we demonstrate that these benefits translate to learning. We begin with the simple teacher-student setting introduced in Section~\ref{sec:theory}, using prospective neurons to enable teaching signal synchronization across a variety of learning rules. This controlled environment allows us to study how non-ideal prospective neurons affect teaching signal synchronization and learning performance. We then investigate how prospective neurons support reward-based learning in a control task and finally combine them with leaky neurons to successfully solve a reaching task that integrates both working memory and motor control.

\subsection{Prospective neurons support teaching signal synchronization in a large variety of learning rules}
\label{subsec:learning-teacher-student}

Our first series of experiments investigates how imprecise tracking impacts online learning performance in feedforward networks and applies prospective dynamics to more biologically plausible learning rules. We employ a teacher-student setup similar to our previous theoretical analysis, with one key difference: synaptic plasticity is now activated. We continue to use the continuous-time backpropagation of errors from (\ref{eqn:continuous-time-BP}) as our default learning rule, but with $W$ evolving over time according to the following dynamics:
\begin{equation}
    \label{eq:weight_update_eq_bp}
    \tau_W \dot{W}_t = -\delta_t \rho(u_t)^\top,
\end{equation}
and using approximate prospective dynamics (adaptive neurons or $\tau' \neq \tau$) whenever needed.
The network therefore learns online, without buffering any updates. As discussed in Section~\ref{subsec:synchro_as_tracking}, these parameter updates follow the gradient when the network is at equilibrium. To isolate the impact of individual components of the learning rule from sampling noise and reduce result variance, we use several samples in parallel (batch size of $50$). While this approach is not fully online, we relax this constraint in subsequent sections to train neural networks in a purely online manner.

\paragraph{Better tracking leads to better online learning performance.} 
Throughout this paper, we have argued that learning performance and  tracking the target trajectory of neural activities are intimately related through teaching signal synchronization. We now verify this relationship quantitatively.

We find that prospective dynamics achieve learning performance equivalent to the instantaneous backpropagation of errors algorithm, where activities and error signals are computed instantaneously at each time step, following Equations~\ref{eqn:BP-forward} and \ref{eqn:BP-backward}. This finding holds across a wide range of $\tau$ values, as shown in Figure~\ref{fig:teacher_student}B. Notably, even in challenging conditions where the sequence length is 5 times the value of $\tau$ -- preventing the dynamics from reaching perfect tracking -- prospective dynamics still achieve competitive performance.

Our results demonstrate that tracking precision strongly correlates with online learning performance. By examining the effects of leaky dynamics, adaptive dynamics, and time constant mismatch in prospective dynamics (Figure~\ref{fig:teacher_student}B, D, and F), we observe that in all cases, learning performance approaches that of instantaneous backpropagation as the tracking error converges to $0$. This confirms that the dependencies on data and network properties highlighted by our theory are qualitatively the same regarding learning performance.

\begin{table}[t]
    \centering
    \small
    \begin{tabular}{llll}
            \toprule
            \textbf{Connectivity} & \textbf{Method} & \textbf{Train loss} & \textbf{Test loss}\\
            \midrule
            Feedforward & Instantaneous BP & $3.94 \times 10^{-2}$ & $1.15 \times 10^{-3}$ \\
            \midrule
            Feedforward & Prospective BP & $4.16 \times 10^{-2}$ & $1.13 \times 10^{-3}$ \\
            Feedforward & Prospective DFA & $5.97 \times 10^{-1}$ & $1.32 \times 10^{-2}$ \\
            Feedforward & Prospective FA & $8.08 \times 10^{-1}$ & $1.75 \times 10^{-2}$ \\
            \midrule
            Feedforward & Leaky BP & $7.08 \times 10^1$ & $1.32 \times 10^1$\\
            \midrule
            \midrule
            Recurrent & Prospective RBP & $1.85 \times 10^{-3}$ & $8.34 \times 10^{-4}$ \\
            Recurrent & Prospective hEP & $1.53 \times 10^{-2}$ & $1.10 \times 10^{-2}$\\
            \midrule
            Recurrent & Leaky RBP & $1.94 \times 10^{-2}$ & $1.30 \times 10^{-2}$ \\
            Recurrent & Leaky hEP & $1.36 \times 10^1$ & $3.39 \times 10^1$ \\
            \bottomrule
        \end{tabular}
        \vspace{0.4cm}
        \caption{\textbf{Comparison of different learning algorithms on the teacher-student learning task of Section~\ref{subsec:learning-teacher-student}.} The first group of learning rules are used to train feedforward neural networks and the second one neural networks with recurrent connections relaxed to equilibrium. Prospective dynamics greatly improve learning performance, to the point of matching the performance of instantaneous backpropagation, which is explained by better tracking abilities. They additionally apply to a large variety of learning rules, such as (direct) feedback alignment or holomorphic equilibrium propagation. The former requires explicit error representation while the latter leverages oscillations to implicitly compute them. (R)BP stands for (recurrent) backpropagation, (D)FA for (direct) feedback alignment, and hEP for holomorphic equilibrium propagation.}
        \label{tab:learning_experiments}
\end{table}

\paragraph{Prospective dynamics applied to different learning rules.} To demonstrate the generality of the notion of teaching signal synchronization as tracking, we show that prospective dynamics enable effective online learning across various learning rules beyond the backpropagation method we have considered so far.

Backpropagation suffers from the weight transport problem \citep{grossberg_competitive_1987}: its feedback pathway is symmetric to the forward one, which is biologically implausible. Algorithms such as feedback alignment \citep{lillicrap_random_2016} and direct feedback alignment \citep{nokland_direct_2016} were developed to address this limitation by using distinct random feedback pathways to backpropagate errors. We reformulate these rules using continuous-time dynamics similar to (\ref{eqn:continuous-time-BP}) and present their performance in Table~\ref{tab:learning_experiments}. While they do not match the performance of backpropagation, they achieve significant results and outperform backpropagation with leaky dynamics, which uses symmetric weights.

The learning rules discussed so far require two key elements: neural activities and error signals. Several studies suggest that apical dendrites within biological neurons may encode these error terms \citep{sacramento_dendritic_2018, richards_dendritic_2019, mikulasch_where_2023}, although conclusive empirical evidence remains elusive. An alternative class of models uses variations in neural activity from external feedback to implicitly evaluate error signals \citep{ackley_learning_1985, baldi_contrastive_1991, movellan_contrastive_1991, scellier_equilibrium_2017}, typically requiring two distinct phases: one with a teaching signal and one without. To overcome this constraint, early work by \citet{baldi_contrastive_1991} proposed that teaching signal oscillations could enable single-phase learning without directly representing errors. In this approach, feedback remains continuously active with ongoing variations, eliminating the need for discrete phases. However, this introduces a new timescale for the oscillating signal, which must be sufficiently slow to allow neurons to track the target trajectory but faster than the input sequences to provide meaningful signals. Prospective dynamics solve this challenge by enabling neurons to effectively track feedback oscillations, regardless of their intrinsic timescales, potentially providing a key mechanism for contrastive learning in the brain. Using the holomorphic equilibrium propagation framework \cite[hEP]{laborieux_holomorphic_2022}, we demonstrate that prospective dynamics indeed support contrastive learning, with oscillation-based methods performing comparably to instantaneous backpropagation (see Table~\ref{tab:learning_experiments}), without requiring explicit error signal representation. Note that the model without a teaching signal still uses feedback connections and is therefore not a feedforward neural network. This kind of model is known as a deep equilibrium model in the machine learning literature \citep{bai_deep_2019} and needs to be trained with recurrent backpropagation \citep[][RBP]{almeida_backpropagation_1989, pineda_recurrent_1989}. Compared to (\ref{eqn:continuous-time-BP}), the learning and neural dynamics hardly change, as shown in Appendix~\ref{app:theoretical_derivations}.

Finally, we highlight that numerous additional learning rules requiring neural activity to be at equilibrium can benefit from prospective dynamics. These include predictive coding \citep{whittington_approximation_2017} (which yields the neuronal least-action principle of \cite{senn_neuronal_2024}) and control-based learning \citep{meulemans_least-control_2022}. 

\subsection{Prospective neurons support reward-based learning in control tasks}
\label{subsec:cartpole}

\begin{figure}[!t]
    \centering
    \includegraphics[scale=1.2]{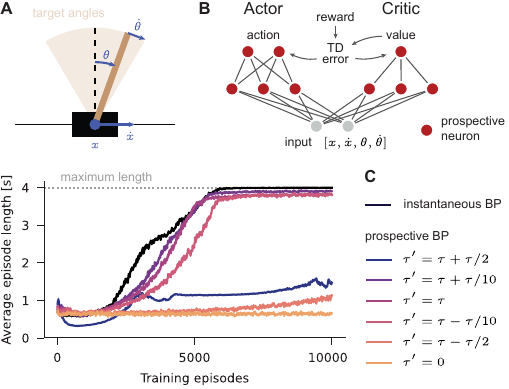}
    \caption{\textbf{Prospective dynamics support learning in a control task.} \textbf{A.} Visual depiction of the inverted-pendulum task, in which the goal is to balance the pole in an upright position. Reward rate is $1$ as long as the pole remains in the ``target angles'' region and that cart is not too far from the origin. \textbf{B.} The current state, comprising the position $x$ and velocity $\dot{x}$ of the cart, and the angle $\theta$ and angular velocity $\dot{\theta}$ of the pole, is fed to an actor-critic network. This network outputs an action (left or right) and an estimated value representing the expected discounted reward. At each time step, the estimated value is combined with the reward to compute the temporal difference (TD) error, which then serves as a teaching signal for both the actor and the critic. \textbf{C.} The prospective backpropagation algorithm solves the task -- maintaining the pole in the target region for $4$s -- as efficiently as its instantaneous counterpart. Gradually reducing the prospective component of dynamics ($\tau' \rightarrow 0$ with $\tau'$ as in Section~\ref{subsec:time_constant_mismatch}) diminishes learning capability, eventually preventing learning altogether. Increasing $\tau'$ has similar effects. Small deviations around $\tau' = \tau$ do not significantly affect performance, as seen with the $\tau' = \tau - \tau/10$ and $\tau'= \tau + \tau / 10$ curves. We use $\tau = 100\mathrm{ms}$ and report the rolling average of episode length using a $30$-episode window, averaged across $5$ seeds. Standard deviations (approximately $500$ms) are omitted to avoid visual clutter, with no significant differences observed between methods. See Methods section for implementation details.}
    \label{fig:cartpole}
\end{figure}

Having demonstrated the benefits of prospective neurons in simplified settings, we now examine their performance in a behaviorally relevant motor control task: the inverted-pendulum (also known as \texttt{Cartpole}, Figure~\ref{fig:cartpole}A). In this task, the agent must balance a pole in an upright position for as long as possible, given information about the current cart position, cart velocity, pole angle, and pole angular velocity. The agent receives a reward only when the pole remains within a designated target region.

While this task is typically straightforward for reinforcement learning algorithms in its standard form, we investigate a more challenging variant with three key modifications: First, we operate in the near continuous-time regime (simulation timestep $\Delta t = 1$ms), where estimating the expected discounted reward becomes notably difficult \citep{tallec_making_2019}. Second, we implement purely online learning without leveraging any offline replay or trajectory batching. Third, since the agent's actions influence future states, delays become even more detrimental than in our previous experiments. Given the Markovian nature of the environment, we implemented the agent as a memory-less feedforward neural network trained with an online continuous version of the advantage actor-critic algorithm \citep{doya_reinforcement_2000}. The algorithm and the choice of parameters are discussed in Appendix~\ref{app:additional_details_pendulum}.

Figure~\ref{fig:cartpole} illustrates the advantage of prospective dynamics (over leaky dynamics) in this task. Most notably, adaptive dynamics enable rapid error signal delivery: the prospective backpropagation algorithm ($\tau = 100$ms) performs comparably to its instantaneous counterpart.
We further investigate how the time-constant mismatch of the prospective component affects learning performance ($\tau' \neq \tau$). For this analysis, we simulate the dynamics described in Section~\ref{subsec:time_constant_mismatch}, varying $\tau'$ between $0\tau$ and $1.5\tau$. The learning trajectories for various $\tau'$ values in Figure~\ref{fig:cartpole}C demonstrate a clear pattern: as dynamics approach those of a leaky integrator ($\tau' = 0\tau$), performance deteriorates progressively until the agent completely fails to learn the task (occurring around $\tau' \approx \tau / 2$). This confirms that prospective dynamics are essential for successful task completion. We also note that the agent is still able to learn to solve the task almost perfectly under a relatively small time-constant mismatch on both sides ($\tau'=\tau + \tau/10$, $\tau'=\tau-\tau/10$). This is consistent with the teacher-student tracking example in Fig~\ref{fig:teacher_student}.

\subsection{Prospective neurons support the learning of memory-storing non-prospective neurons}
\label{subsec:rnn}

\begin{figure*}
    \centering
    \includegraphics[scale=1.2]{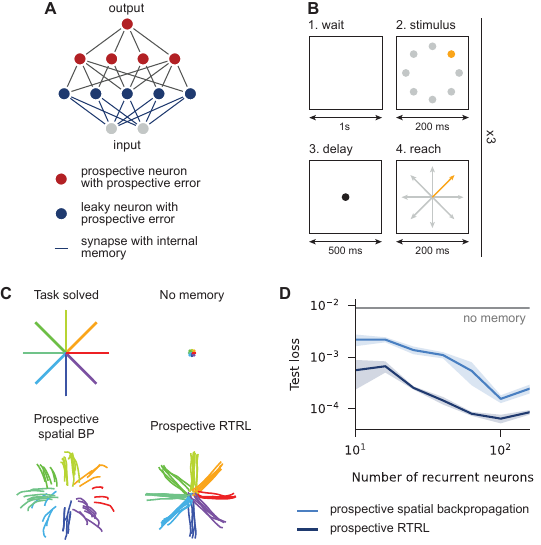}
    \caption{\textbf{Prospective neurons support learning of memory-storing non-prospective leaky neurons.} \textbf{A.} We consider a multi-layer neural network combining complex-valued leaky neurons that integrate information over time (first layer) with prospective neurons (in the subsequent two layers). All neurons have prospective errors, following the prospective version of the $\delta$ dynamics in (\ref{eqn:continuous-time-BP}). The first layer's role is to store memories of past inputs, while the subsequent layers provide non-linear and instantaneous processing of these memories. \textbf{B.} We train the network on a delayed reaching task requiring production of context-dependent target trajectories. \textbf{C.} The network successfully solves the task when trained online with prospective real-time recurrent learning. When memory-specific parameters are frozen (prospective spatial backpropagation), the network can only adjust existing movements but cannot refine them to match the desired behavior. When the leaky neurons are replaced by prospective ones, the network loses its memory capacity and produces no movement as its best response. \textbf{D.} Learning memory-specific parameters becomes increasingly critical when the number of non-prospective leaky neurons is limited. As this number approaches infinity, the pressure decreases since the feedforward network in the last two layers has a wider range of movements to select from to solve the task.}
    \label{fig:rnn}
\end{figure*}

Previous sections demonstrated that prospective neurons enable faster signal propagation through neural networks, resulting in improved learning performance. However, these neurons are inherently memory-less in the prospective limit and thus insufficient for solving complex spatio-temporal learning tasks. To address this limitation, we combine them with non-prospective leaky neurons that maintain a memory of past inputs, training the resulting hybrid network with a continuous time version of the online real-time recurrent learning algorithm \cite[RTRL]{williams_learning_1989}.

Specifically, we implement a network architecture consisting of one layer of complex-valued non-prospective linear leaky neurons followed by two feedforward layers of prospective neurons, as illustrated in Figure~\ref{fig:rnn}A. The non-prospective leaky neurons operate without recurrent connections between them, making exact online gradient calculation tractable \citep{zucchet_online_2023}. Using complex-valued neurons (instead of real-valued) ensures that the network remains expressive as a dense linear recurrent layer \cite{orvieto_resurrecting_2023}. Online gradient calculation requires maintaining one internal state per incoming synapse of the leaky neurons -- conceptually similar to an eligibility trace \citep{zenke_superspike_2018, bellec_solution_2020} -- and one per such neuron. We refer to these states as parameter sensitivities $s^{\theta}_t$, which mathematically correspond to the derivative $s_t^\theta := \mathrm{d}_\theta h_t$, where $h_t$ represents the state of the leaky neurons. These sensitivities can then be combined with error signals $\delta$ (similar to those in instantaneous spatial backpropagation of errors through the network hierarchy) to estimate the gradient. For example, the update for the time constants $\tau$ of the leaky neurons is given by:
\begin{equation}
    \tau_\tau \dot{\tau}_{t} = s_t^\tau \odot \delta_t
\end{equation}
where $\odot$ denotes the element-wise product. This weight update follows the gradient accurately when error signals $\delta_t$ arrive in sync. To ensure this synchronization, the use of prospective neurons and errors for the final two layers is critical, as well as having prospective errors associated with each leaky neuron. Further architectural details and learning rule derivations are provided in the Methods section.

\ifpreprint
    \vspace{-0.1cm}
\fi

Despite its relative simplicity, this network architecture can learn non-trivial tasks in a purely online manner. We demonstrate this capability using a delayed reaching task (Figure~\ref{fig:rnn}B) where the network receives a stimulus indicating a reaching direction, must wait for a ``Go" cue, and then move a virtual arm in the corresponding direction. Figure~\ref{fig:rnn}C illustrates how different methods perform on this task after processing $500$ sequences. With the learning rule described above, the network achieves near-perfect task performance. When memory is removed (e.g., by replacing leaky neurons with prospective ones), the network produces no movement after learning, as it cannot retain the directional information needed for the delayed response. Alternatively, when leaky neurons are present but their time constants remain fixed (prospective spatial BP; similar to reservoir computing \citep{lukosevicius_reservoir_2009}), the network learns to select and modify existing movement patterns but cannot reshape them to produce precisely targeted movements. This limitation becomes particularly pronounced in networks with fewer non-prospective neurons (Figure~\ref{fig:rnn}D), where the available variety of movement is more restricted.
Notably, with our chosen time constant ($\tau=100$ms), we found that networks using standard leaky dynamics (not shown in the figure) instead of prospective dynamics do not outperform the memory-less baseline -- highlighting that teaching signal synchronization is essential for effective spatio-temporal learning.

\ifpreprint
    \vspace{-0.1cm}
\fi

We conclude by noting that this kind of architecture, while exhibiting rudimentary recurrence patterns, has recently demonstrated remarkable power for sequence learning when stacked hierarchically \citep{gu_efficiently_2022, orvieto_resurrecting_2023}. The online learning rule we explore naturally extends to these deeper networks \citep{zucchet_online_2023}, albeit with minor approximations to the gradient. We thus expect it to be able to tackle more challenging tasks.
\section{Discussion}

\paragraph{Connections to prospective dynamics.} Dynamics similar to the prospective dynamics studied in this paper have been examined in various fields, such as neuroscience, optimization theory, and statistical estimation. Physiological evidence \citep{ulrich_dendritic_2002} has highlighted the ability of biological neurons to phase-advance their output. Drawing from these findings, \citet{mi_spike_2014} studied spike frequency adaptation as a potential mechanism for tracking and phase-advancing traveling waves in attractor networks. The neuronal least action \citep{senn_neuronal_2024} and latent equilibrium \citep{haider_latent_2021} theories propose that this prospective ability may enable neural activity to respond instantaneously to external stimuli. \cite{ellenberger_backpropagation_2024} explore how time constant mismatches of the type we consider in Section~\ref{subsec:time_constant_mismatch} could support temporal credit assignment. The dynamics we study are a generalization of those of the neuronal least action theory beyond the energy-based systems it considers. Our dynamics are also a generalization of the prediction-correction algorithm \citep{zhao_novel_1998, simonetto_prediction-correction_2017} in time-varying optimization \citep{simonetto_time-varying_2020}. This problem arises in many fields, such as adaptive control \citep{landau_adaptive_2011} and online learning \citep{zinkevich_online_2003}, and it consists of finding a trajectory of parameters that minimizes a time-varying cost function. The prediction-correction algorithm consists of two parts: the prediction part of the dynamics, which aims to keep the tracking error constant, and a correction part that pushes the current parameters towards the current minimizer. The mathematical details of these two connections are provided in the Methods section. Finally, our prospective dynamics have interesting connections to control and estimation theory. When considering the input received by a neuron as an error term, the prospective input is akin to a proportional-derivative controller (PD) \citep{minorsky_directional_1922} that steers neural activity. The prediction-correction view is also reminiscent of Kalman filtering \citep{kalman_new_1960}, in which the current state estimate is updated using a model of the true dynamics of the system and corrected using the current observation. Instead of being explicitly defined through a differential equation, as in the Kalman filter, the target is implicitly defined as an equilibrium point here.

\paragraph{Mechanisms to deal with the delays inherent to neural computation.} Prospective neurons add to the list of mechanisms that can reduce delays in biological neural networks. At the population level, \cite{knight_dynamics_1972} and \cite{van_vreeswijk_chaos_1996, van_vreeswijk_chaotic_1998} have shown that groups of neurons can respond to external inputs faster than the characteristic time constant of individual neurons. When it comes to learning, this implies that teaching signals could arrive in sync at the population level, which is not precise enough to support the learning of hierarchical networks. At the neuronal level, eligibility traces \citep{gerstner_eligibility_2018} have been argued to support temporal credit assignment \citep{zenke_superspike_2018, bellec_solution_2020}. Intuitively, the role of these traces is to apply a similar delay to the pre-synaptic term in the learning rule as that of the post-synaptic term, which is the error signal in our case. This can sometimes recover the true gradient signal for neurons that are not recurrently connected \citep[e.g.,][]{mozer_focused_1989}, assuming the downstream processing is instantaneous. However, the deeper the downstream network, the harder it becomes to mimic the delay it induces on the error with a simple fading memory. As a consequence, there will be a temporal mismatch between the pre- and post-synaptic terms of the learning process, which will be particularly pronounced for remote layers, and learning will be significantly affected. Prospective neurons constitute a unique mechanism in their ability to enable the precise teaching signals required for learning deep networks.

\paragraph{Neural implementation of adaptive dynamics.} We have thoroughly discussed the theoretical properties of prospective neurons in Section~\ref{sec:theory} and have shown that adaptive dynamics constitute a more realistic implementation while maintaining most of the important prospective properties. Critical to the theory of Section~\ref{sec:physical_constraints} and its experimental counterpart is the hypothesis that the adaptation time constant $\tau_a$ is faster than the membrane time constant $\tau$. While this is outside the regime in which adaptive neurons traditionally lie \citep{brette_adaptive_2005}, we argue that mechanisms that could support fast adaptation exist. For instance, it might be caused by the fast inactivation of sodium currents involved in generating action potentials. In this case, the estimate of the input temporal derivative is interpreted as a sodium current. It has an ultra-fast activation in the range of $1\,$ms and a fast inactivation in the range of $10\,$ms (respectively, the gating variables $m$ and $h$ in the Hodgkin-Huxley model \citep{hodgkin_quantitative_1952,gerstner_neuronal_2014}). Although the sodium current involves a product of these voltage-dependent gating variables, their overall effect can be approximated as the sum of an ultra-fast voltage-dependent excitation ($f$) and fast inhibition ($a$) that, together, cause a prospective voltage drive \citep{brandt_prospective_2024}.
\section*{Methods}

\paragraph{Link with neuronal least action (NLA) principle \citep{senn_neuronal_2024} latent equilibrium (LE) \citep{haider_latent_2021} theories.} These two theories both study how prospective neurons theoretically enable instantaneous information propagation in energy-based models that rely on prediction errors. The NLA derives the dynamics from a least action principle, while the LE minimizes energy. The difference between these theories relies on which quantities are prospective: input currents/firing rates in the NLA and voltages in the LE. Our theory is a generalization of the NLA to more general dynamical systems; we recover the NLA by having $f_\theta(s, t) = s - \nabla_s E_\theta(s, t)$ with $E$ an energy function. This has important consequences: First, on a conceptual level, it helps decouple the role of prospective dynamics from other modeling choices. Second, it simplifies the mathematical analysis of the dynamics, especially in non-ideal settings. Third, and perhaps more importantly, it enables the decoupling of neural activity from error signals. In particular, activity can be non-prospective when errors are prospective. This is crucial for combining prospective neurons with the RTRL rule we use in Section~\ref{subsec:rnn}.

\paragraph{Link with the prediction-correction algorithm.}  The prediction-correction algorithms aim to solve a time-varying optimization problem $\min_s E(s, t)$ for all $t$, using the solutions obtained from previous timesteps. Assuming $f(s, t)$ to be $s - \nabla_s E(s, t)$, the prediction-correction algorithm \citep{simonetto_prediction-correction_2017} minimizes the cost $E$ through
\begin{equation}
    \label{eqn:prediction_correction}
    \tau \dot{s}_t = \left (\mathrm{Id} - \frac{\partial f}{\partial s}(s_t, t) \right )^{-1} \left ( -s_t + f(s_t, t) + \tau \frac{\partial f}{\partial t}(s_t, t)\right )
\end{equation}
In this setting, $s$ does not represent any neural activity but rather the parameters of the system that are to-be-optimized. While different from the prospective dynamics at first glance, both dynamics are the same. To make this link, we can rewrite the prospective dynamics of Eq.~\ref{eqn:prospective_dynamics} as $\mathrm{d}_t \left [ s_t - f_\theta(s_t, t) \right ] = - \left [ s_t - f_\theta(s_t, t) \right ]$. Then, we apply the chain rule to the left-hand side term of the previous equation and invert the obtained Jacobian to get the desired result. The correction part of~(\ref{eqn:prediction_correction}), 
\begin{equation}
    \tau \dot{s} = \left ( \mathrm{Id} - \partial_s f(s_t, t) \right )^{-1} (-s_t+ f(s_t, t)),
\end{equation}
is the continuous-time version of Newton's algorithm for finding a solution of $s - f(s, t)$ as if $t$ is fixed. On the other side, the prediction part, $\tau\dot{s}_t = \tau(\mathrm{Id} - \partial_s f(s_t, t))^{-1} \partial_t f(s_t, t)$, aims to keep the error $s_t- f(s_t, t)$ constant and arises from the implicit function theorem \citep{dontchev_implicit_2009}. The prospective dynamics we study can thus also be understood from this perspective.

\paragraph{Simulation of prospective dynamics.} Several challenges arise when simulating prospective dynamics on digital computers. On one side, numerical integration of the prospective dynamics (\ref{eqn:prospective_dynamics}) requires schemes suited for implicit differential equations, as both left- and right-hand sides depend on $\dot{s}$. Such schemes do not benefit from the extensive theoretical understanding that ordinary differential equations do. On the other side, simulating the equivalent prediction-correction dynamics (\ref{eqn:prediction_correction}) requires inverting the Jacobian matrix $(\mathrm{Id} - \partial_s f(s_t, t))$ at every time step, which is notoriously expensive. We found that the following Euler-like implicit integration scheme works well for our purposes: we update $s$ by approximating $\dot{s}_t$ as $\Delta t^{-1} [s_{t+\Delta t} - s_t]$ and approximate $\mathrm{d}_t [f_\theta(s_t, t)]$ as $\Delta t^{-1} [f_\theta(s_t, t) - f_\theta(s_{t-\Delta t}, t-\Delta t)]$, with $\Delta t$ being the step size of the integration scheme. That is,
\begin{equation}
    s_{t+\Delta t} = s_{t} + \frac{\Delta t}{\tau} ( -s_t + f_\theta(s_t, t) ) + f_\theta(s_t, t) - f_\theta(s_{t-\Delta t}).
\end{equation}
While there exists a flourishing literature studying the theoretical properties of numerical schemes for differential algebraic equations \citep[e.g.][]{marz_numerical_1992}, of which the prospective dynamics is an example, we are not aware of any theoretical guaranties for this specific integration scheme. Yet, we found that it performed similarly to the Euler integration of the prediction-correction dynamics, while ours has a much lower computational footprint.

\paragraph{Experimental setup for Sections~\ref{sec:theory} and \ref{sec:physical_constraints}.}
We use a feedforward neural network with 2 hidden layers of size $5 - 20 - 20 - 1$ and a sigmoid non-linearity. The weights of the network are drawn from a normal distribution with variance $1 / n_\mathrm{out}$ and $n_\mathrm{out}$ representing the number of output neurons of a given layer. The teacher network has the same size, but its weights are drawn from a distribution with a standard deviation 3 times larger. Each of the 5 features of the input $x$ to the network is a random linear combination of $1000$ sine waves with phases $\varphi_i$ uniformly sampled from $[0, 2\pi]$ and angular velocities $\omega_i$ from $[\omega_0, 2\omega_0]$. More precisely, we have
\begin{equation}
    (x_t)_i = \sum_{j=1}^{1000} M_{ij} (s_{t})_j ~~ \mathrm{with} ~ (s_{t})_j = \sin(\omega_j t + \varphi_j) ~ \text{and} ~ M_{ij} \sim \mathcal{N}\left(0, \frac{1}{\sqrt{d_\mathrm{input}}} \right )
\end{equation}
The target output $y_t$ used in the dynamics (\ref{eqn:continuous-time-BP}) is the instantaneous processing of $x_t$ by the teacher network.

In Figure~\ref{fig:theory}A, we use the Euler integration scheme with a step size of $\Delta t = 10$ms for 1000 sequences of length $T=100s$. The different weights are fixed throughout this experiment; we only vary $\tau$ (its logarithm is drawn uniformly from $[\log 0.05, \log 1]$) and $\omega_0$ (its logarithm is drawn uniformly from $[\log 0.01, \log 0.2]$). We additionally resample all the sine waves for each sequence. The metric we report is the $\max_{t \geq 25s} \lVert s_t - f_\theta(s_t, t) \rVert$.

In Figure~\ref{fig:theory}B, we use the integration scheme detailed above with a step size of $0.05\,$ms. We take $\omega_0$ to be $2\pi / 10$. The metric we report is $\lVert s_t - f_\theta(s_t, t) \rVert$.

\paragraph{Experimental setup for Section~\ref{subsec:learning-teacher-student}.} The experimental setup is almost the same as the one for the tracking measurements described above. There are a few additional changes: we use a student of size $10 - 100 - 100 - 5$ and a teacher of size $10 - 20 - 20 - 5$, both with ReLU non-linearities. The input combines $100$ different sine waves with a characteristic angular velocity of $\omega_0 = 2\pi / 10$, and the step size $\Delta t$ is equal to $1$ ms. By default, $\tau$ is taken to be equal to $1s$, and each sequence is of length $5s$. Note that we use shorter time sequences than in the tracking experiments here to reduce temporal correlations in the data and thus make gradient-based learning easier. We train the networks over $100$ epochs, using the Adam optimizer with a learning rate $\Delta t / \tau_W$ and a cosine scheduler. We test the learned network on a held-out set of $5000$ sequences at the end of learning, keeping the same neural dynamics as those during learning. The network weights are updated after each time step. The learning rate is tuned independently for each method using a grid search over $[0.001, 0.003, 0.01, 0.03, 0.1, 0.3, 1]$ and 2 seeds. For the plots with varying parameters (Figure~\ref{fig:teacher_student}B, D and E), we pick the learning rate yielding the best average test loss.

For the recurrent BP and holomorphic EP learning experiments (Figure~\ref{tab:learning_experiments}), we used the same setup as above, with the following differences: The architecture of the student network is a continuous Hopfield network with a shifted sigmoid $z \mapsto 1/(1+e^{-4z+2})$ instead of ReLU. 
The teacher network is still a feed-forward network using the sigmoid activation instead of ReLU for faster target computation.
Both the teacher and student networks have the same dimensions as in the instantaneous BP experiment.
For holomorphic EP, the complex teaching oscillations $\beta = |\beta|e^{{\rm i}\omega_\beta t}$ have an amplitude of $|\beta|=0.2$, and a pulsation of $\omega_\beta = 190\pi$ rad/s. The error is decoded from the online neuronal oscillations using an exponential moving average filter on a timescale of $\tau_{\rm ema} = 31.5$ms. The best parameters for $|\beta|$, $\omega_\beta$, and $\tau_{\rm ema}$ are found through grid search.
The training lasts for $200$ epochs. 
For more details on the theory of holomorphic EP, see Appendix~\ref{app:hep}.

\paragraph{Experimental setup for Section~\ref{subsec:cartpole}.} The simulation of the inverted pendulum is done with Euler integration ($\Delta t = 0.001$s). At each timestep, the agent can take one of two actions that exert a force either to the left or to the right. An episode ends if either the maximum timestep is reached ($T=4$s) or if the pole falls outside an admissible range ($-2.4 \leq x \leq 2.4, -12^\circ \leq \theta \leq 12^\circ$). We note that since prospective tracking takes time to react to discontinuity, we restrict the agent to receive a reward of $1$ only when the pole is balancing within a smaller range ($-2 \leq x \leq 2, -8.5^\circ \leq \theta \leq 8.5^\circ$) than the true admissible range and receives a reward of 0 otherwise. This change was critical for the value function estimator and its gradient to converge and learn the end of the episode properly. We parameterize the actor and critic with one layer neural networks consisting of $64$ hidden neurons and ReLU activation. We train the two networks for $10000$ episodes with a continuous advantage actor-critic (A2C) loss \citep{doya_reinforcement_2000, mnih_asynchronous_2016}. The networks are optimized with stochastic gradient descent with early stopping, where the gradient is calculated with prospective backpropagation dynamics. The learning rate, $\Delta t / \tau_W$, is tuned for prospective BP for each $\tau'$s in the range [$10^{-7}$, $10^{-6}$, $10^{-5}$, $10^{-4}$, $10^{-3}$]. We note that within the hyperparameter setups where training is effective, the learning rate did not affect the speed of learning in prospective BP with different $\tau'$. A detailed description of the full online prospective A2C training algorithm is provided in the Appendix~\ref{app:additional_details_pendulum}. 

\paragraph{Learning rule for the network of Section~\ref{subsec:rnn}.} In this series of experiments, we study a network comprising one layer of complex-valued leaky neurons, followed by a one hidden-layer feedforward neural network. The latter is exactly the same as what we use in the rest of the paper, including how error signals evolve, and it follows the dynamics of the prospective backpropagation algorithm. We will thus assume that it is always tracking the target trajectory. The complex-valued leaky neurons followed the dynamics:
\begin{equation}
    \tau \dot{u}_t = -(1 + i \omega)u_t + Wx_t.
\end{equation}
Note that here $u$ and $W$ are complex-valued terms. One can think of $u$ as a linearized version of a complex neuron, e.g. \citep{izhikevich_resonate-and-fire_2001, schaffer_complex-valued_2013}. We do not enter further biophysical implementation details and focus on the learning dynamics instead.

We can now compute the loss of the gradient with respect to some of the parameters $\theta$ of the leaky neuron
\begin{equation}
    \frac{\mathrm{d}L}{\mathrm{d}\theta} = \int \frac{\mathrm{d}L_t}{\mathrm{d}\theta}\mathrm{d}t = \int \frac{\partial L_t}{\partial u_t} \frac{\mathrm{d} u_t}{\mathrm{d}\theta} \mathrm{d}t
\end{equation}
where we used that $u_t$ only directly affects $L_t$ in the second equation. The term $\partial_{u_t} L_t$ is exactly the $\delta_t$, up to a transpose, of (spatial) backpropagation, so it can be made accessible by having a prospective error ``neuron" associated with each leaky neuron. We are left with estimating $\mathrm{d}_\theta u_t$. As those leaky neurons are independent of each other, the parameters of one do not affect the state of the others, so we can ignore all the $\mathrm{d}_{\theta_i} u_{t,j}$ for $i \neq j$; with a slight abuse of notation, we get rid of the identically zero terms in $\mathrm{d}_\theta u_t$ in the following. Importantly, it implies that the non-zero terms of $\mathrm{d}_\theta u_t$ are exactly the size of $\theta$, so we can store this extra state in the neuron for $\tau$ and $\omega$, and in the synapse for $W$. We have
\begin{align}
    \tau \dot{\frac{\mathrm{d} u_t}{\mathrm{d}\tau}} &= (-\mathbf{1} + i \omega) \frac{\mathrm{d} u_t}{\mathrm{d}\tau} - \tau^{-1} \left ( (-\mathbf{1} + i\omega) u_t + Wx_t \right) \\
    \tau \dot{\frac{\mathrm{d} u_t}{\mathrm{d}\omega}} &= (-\mathbf{1} + i \omega) \frac{\mathrm{d} u_t}{\mathrm{d}\omega} + i u_t\\
    \tau \dot{\frac{\mathrm{d} u_t}{\mathrm{d} W}} &= \mathrm{diag} [-1 + i \omega]  \odot  \frac{\mathrm{d} u_t}{\mathrm{d}W} + x_t \mathbf{1}^\top
\end{align}
In the last equation, we used $\mathbf{1}$ to denote a vector of ones of appropriate size. Finally, updating those parameters according to
\begin{align}
    \tau_\tau \dot{\tau}_t &= -\frac{\mathrm{d}u_t}{\mathrm{d} \tau} \odot \delta_t\\
    \tau_\omega \dot{\omega}_t &= -\frac{\mathrm{d}u_t}{\mathrm{d} \omega} \odot \delta_t\\
    \tau_W \dot{W}_{t, ij} &= -\frac{\mathrm{d}u_t}{\mathrm{d} W_{ij}} \delta_{t, j}
\end{align}
will follow the instantaneous gradient. Note that if parameter updates are buffered and not applied, the estimated gradient for all parameters will exactly match the gradient, that is typically computed with backpropagation-through-time. In short, we show that if each neuron and synapse has an additional hidden state that evolves with the equations above, exact online gradient calculation is possible for the kind of neural networks we study here.  In our experiments, it will not be entirely true as plasticity is always on and as we initialize neural activity at a default state that is not yet on the target trajectory.

\paragraph{Experimental setup for Section~\ref{subsec:rnn}} Each sequence of the delayed reach task consists of 3 subsequences of the form ``Wait", ``Stimulus", ``Delay" and ``Reach" which last respectively $1$s, $200$ms, $500$ms and $200$ms. A sequence is thus $5.7$s. We use $\Delta t = 1$ms. The input given to the network has $9$ dimensions, $8$ for the one-hot encoded version of each stimulus, and $1$ for the ``Go" cue in the ``Reach" phase. The target movement, which is 2-dimensional, is a line from the center to a location encoded by the stimuli. We use the mean squared distance to the target movement as the loss to train the network. The network consists of one layer of leaky neurons and two layers of prospective neurons. We fix the hidden layer of the multi-layer perceptron to have $100$ neurons and vary the number of leaky neurons from $10$ to $158$ in Figure~\ref{fig:rnn}. The time constants of the prospective neurons ($u$ and $\delta$) are fixed at $100$ms. The logs of the time constants of the leaky neurons are randomly initialized uniformly from $[\log 100$ms, $\log 10$s$]$. In the prospective RTRL algorithm, the multilayer perceptron is trained using the prospective backpropagation algorithm, and the leaky neurons are trained with the learning rule described above. We train all networks on $500$ sequences without any replay, which more or less corresponds to all possible stimulus combinations, using the Adam optimizer and a cosine learning rate scheduler with an initial learning rate of $0.01$, obtained after a grid search on both methods. Note that it corresponds to having a characteristic time constant of $10$s for all learned parameters. We then measure the test loss on 500 sequences that may have been seen during training. The goal of this loss is not to test generalization but rather the optimization abilities of our learning rule.

\ifpreprint
\section*{Acknowledgments}

We thank Jean-Pascal Pfister, Frederico Benitez, and Alexander Meulemans for early discussions on the project, as well as Seijin Kobayashi and Asier Mujika for their advice on reinforcement learning methods. This research was supported by the Swiss National Science Foundation (grant numbers PZ00P3\_186027, PCEFP3\_202981 and TMPFP3\_210282), an ETH Research Grant (ETH-23 21-1), EU’s Horizon Europe Research and Innovation Program (CONVOLVE, grant agreement number 101070374) funded through SERI (ref 1131-52302), and the Novartis Research Foundation. 
\fi

\bibliography{references}

\newpage
\appendix
\section{Theoretical derivations}
\label{app:theoretical_derivations}

\setcounter{thm}{0}

\subsection{Derivation of the error backpropagation dynamics in continuous time}

In Section~\ref{subsec:synchro_as_tracking}, we claimed that the dynamics (\ref{eqn:continuous-time-BP}) compute the error signal of backpropagation. We prove this statement by leveraging the Lagrangian derivation of \citet{lecun_theoretical_1988}. To this extent, we introduce the Lagrangian
\begin{equation}
    \mathcal{L}_\theta(u, \delta, t) = \ell(u, t) + \delta^\top (f_\theta(u, t) - u)
\end{equation}
and want to compute the derivative with respect to $\theta$ of $\ell(u^*_t, t)$ s.t. $u^*_t = f_\theta(u_t^*, t)$. Note that we write the dependency on time $t$ to be consistent with Section~\ref{subsec:synchro_as_tracking}, but we consider it fixed for now. The Lagrange multiplier method provides a way to compute the derivative of the loss: it is equivalent to computing $\mathrm{d}_\theta \mathcal{L}_\theta(u^*_t, \delta^*_t, t)$ with $u_t^* $ and $\delta_t^*$ such that $\partial_u \mathcal{L}_\theta (u_t^*, \delta_t^*, t) = 0$ and $\partial_\delta \mathcal{L}_\theta (u_t^*, \delta_t^*, t) = 0$ as
\begin{align}
    \frac{\mathrm{d}}{\mathrm{d} \theta}\mathcal{L}_\theta(u_t^*, \delta_t^*, t) & = \frac{\partial \mathcal{L}_\theta}{\partial \theta}(u_t^*, \delta^*_t, t) + \frac{\partial \mathcal{L}_\theta}{\partial u}(u_t^*, \delta^*_t, t) \frac{\mathrm{d}u_t^*}{\mathrm{d} \theta} + \frac{\partial \mathcal{L}_\theta}{\partial \delta}(u_t^*, \delta^*_t, t) \frac{\mathrm{d}\delta_t^*}{\mathrm{d} \theta}\\
    &= \frac{\partial \mathcal{L}_\theta}{\partial \theta}(u_t^*, \delta^*_t, t) + 0 + 0 \\
    &= \frac{\partial f_\theta}{\partial \theta}(u_t^*, t)^\top \delta_t^*.
\end{align}
We therefore have a simple way to estimate the gradient, given that $u_t^*$ and $\delta_t^*$ are accessible. In the case of $f_\theta(u, t) = W\rho(u)$ with $u^0_t = x_t$, as in Section~\ref{subsec:synchro_as_tracking}, we have
\begin{equation}
    \frac{\mathrm{d}}{\mathrm{d}\theta} \mathcal{L}_\theta(u_t^*, \delta_t^*, t) =  \rho(u_t^*){\delta_t^*}^\top,
\end{equation}
which is the usual backpropagation update. Note that there is a slight abuse of notation in the previous equation, as this update only applies to the entries of the weight matrix that are not identically zero.

We have justified the $\theta$-update and are now left with computing $u_t^*$ and $\delta_t^*$. One way to reach this equilibrium is to run the coupled dynamical system
\begin{equation}
    \left \{
    \begin{aligned}
        \tau \dot{u}_t &= \frac{\partial \mathcal{L}_\theta}{\partial \delta}^\top(u_t, \delta_t, t)\\
        \tau \dot{\delta}_t &= \frac{\partial \mathcal{L}_\theta}{\partial u}^\top(u_t, \delta_t, t)
    \end{aligned}
    \right .
\end{equation}
that is 
\begin{equation}
    \label{eqn:leaky_bp}
    \left \{
    \begin{aligned}
        \tau \dot{u}_t &= -u_t + W\rho(u_t) ~~\mathrm{and}~~ u_t^* = x_t \\
        \tau \dot{\delta}_t &= -\delta_t + \rho'(u_t) W^\top \delta_t ~~ \mathrm{and} ~~ \delta^L_t = \nabla \ell (u_t^L, y_t)
    \end{aligned}
    \right .
\end{equation}
for our $f_\theta$ of interest. If the time dependency in $f_\theta$ and $L$ can be ignored, the equilibrium points of these equations satisfy the stationary conditions  $\partial_u \mathcal{L}_\theta (u_t^*, \delta_t^*, t) = 0$ and $\partial_\delta \mathcal{L}_\theta (u_t^*, \delta_t^*, t) = 0$, and we can use them to compute the $\theta$-update.

However, in this paper, we are interested in the setting in which the assumption that $f_\theta$ and $l$ are independent of time breaks. Finding neural activity $u$ and error $\delta$ trajectories that satisfy the stationary equations of the Lagrangian becomes an equilibrium tracking problem. Following the leaky error backpropagation dynamics of Eq.~\ref{eqn:leaky_bp} will inevitably lead to delays in neural activity and poor error signal synchronization, as we analyzed in Theorem~\ref{thm:problem_flow}. Instead, we apply the prospective dynamics to always track equilibrium. Finally, we update $\theta$ online, using the neural activity $u_t$ and error signal $\delta_t$ currently available. This yields our \textit{Prospective BP} algorithm:
\begin{equation}
    \left  \{    
    \begin{aligned}
        \tau \dot{u}_t & = -u_t + W_t\left (\rho(u_t) + \tau \dot{\rho(u_t)} \right )\\
        \tau \dot{\delta}_t & = -\delta_t + \rho'(u_t) W_t^\top \delta_t + \tau \frac{\mathrm{d}}{\mathrm{d} t} \left [ \rho'(u_t) W^\top_t \delta_t  \right ]\\
        \tau_W \dot{W}_t& = -\delta_t \rho(u_t)^\top
    \end{aligned}
    \right .
\end{equation}
In our simulations of Section~\ref{sec:experiments}, we use these dynamics with $f_\theta(u, t) = W\rho(u)$ and a feedforward architecture ($W$ lower diagonal), and measure the loss on the last layer of the network.

As a side note, we highlight the links and differences between the algorithm derived above and (recurrent) backpropagation (RBP) equations. In short, it runs the two phases of (R)BP at the same time. Let us explain why. In (R)BP, the first phase consists in finding an equilibrium satisfying
\begin{equation}
    u_t^* = f_\theta(u_t^*, t).
\end{equation}
The resulting equilibrium can be expressed explicitly, as in the feedforward example we used in Section~\ref{sec:theory}, or, in the general case, requires to run dynamics akin to 
\begin{equation}
    \tau_u \dot{u}_t = -u_t + f_\theta(u_t, \theta).
\end{equation}
The second phase requires computing the error signal:
\begin{equation}
    \delta_t^* = \left (\mathrm{Id} - \frac{\partial f_\theta}{\partial u}(u_t^*, t) \right )^{-\top} \nabla_u \,\ell(u_t^*, t).
\end{equation}
Interestingly, this equality can be both derived from the Lagrangian perspective ($\partial_u \mathcal{L}_\theta(u_t^*, \delta^*_t, t) = 0$), or by leveraging the implicit function theorem. It holds both in the feedforward case (acyclic computational graph) or in the recurrent case (cyclic computational graph). When the underlying network is feedforward, it can easily be computed going backward in the network hierarchy \citep{rumelhart_learning_1986}. When the network is recurrent, this is more tedious and this error term is usually computed by solving the linear system:
\begin{equation}
    \left (\mathrm{Id} - \frac{\partial f_\theta}{\partial u}(u_t^*, t) \right )^\top \delta_t^* = \nabla_u \, \ell(u_t^*, t),
\end{equation}
see, e.g., \citep{almeida_backpropagation_1989, pineda_recurrent_1989, zucchet_beyond_2022}. One way to solve it, assuming $f_\theta$ and $L$ to be independent of time, is to run the dynamics
\begin{equation}
    \tau_\delta \dot{\delta}_t = -\delta_t + \frac{\partial f_\theta}{\partial u}(u_t^*, t)^\top \delta_t+ \nabla_u \ell(u_t^*, t)
\end{equation}
One can remark that the $\delta$ dynamics can be run simultaneously to the $u$ one, as $\delta$ does not influence its evolution. This corresponds to the leaky error backpropagation algorithm of (\ref{eqn:leaky_bp}).

To summarize, the continuous-time backpropagation algorithm can be obtained by framing (recurrent) backpropagation as an equilibrium-based algorithm and remarking that the two phases can be run at the same time. This way, we can get rid of one of the shortcomings of backpropagation, specifically the need for signals specifying in which phase we are and the update locking issue. A crucial step needed to derive this algorithm is to remark that the instantaneous rate-based models used in computational neuroscience and machine learning are by essence equilibrium-based.

\subsection{Proof of Theorem 1}

We here restate Theorem~\ref{thm:problem_flow} and demonstrate it.

\begin{thm}
    Let $f_\theta$ be such that the biggest eigenvalue of the symmetric part of the Jacobian $\partial_s f(s, t)$ is always smaller than a constant $1-\mu$ for $\mu>0$. Let $s^\ast_t$ satisfy $s^\ast_t = f_\theta(s^\ast_t, t)$ for all $t$ and $\lVert \dot{s}^\ast_t \rVert \leq \gamma$. Then, the trajectory of states $s_t$ obtained by integrating the leaky dynamics $\tau \dot{s}_t = -s_t + f_\theta(s_t, t)$ verifies
    \begin{equation*}
        \underset{t\rightarrow \infty}{\mathrm{limsup}} ~ \lVert s_t - s^\ast_t \rVert \leq \frac{\gamma\tau}{\mu}.
    \end{equation*}
    Furthermore, this bound is tight, that is we can find a $f$ for which the upper bound is reached.
\end{thm}

\begin{proof}
    Throughout the proof, we drop the $\theta$ subscript in $f_\theta$ for simplicity. We consider the function $g: \alpha \mapsto -s^*_t - \alpha(s_t-s^*_t) + f(s^*_t + \alpha(s_t-s_t^*), t)$ and integrate it between 0 and 1 to obtain the following relationship:
    \begin{align}
        -s_t + f(s_t,t) &= g(1) \\
        & = g(0) + \int_0^1 \frac{\mathrm{d}g}{\mathrm{d}\alpha}(\alpha) \mathrm{d}\alpha\\
        &= 0 + \left [ \int_0^1 \left (-\mathrm{Id} + \frac{\partial f}{\partial s}(s^*_t+\alpha(s_t-s^*_t),t)\right)  \mathrm{d}\alpha \right ] (s_t-s^*_t).
    \end{align}
    In the last line, we used that $s^*_t$ is an equilibrium satisfying $s^*_t = f(s^*_t, t)$ so that $g(0)=0$. It follows that
    \begin{equation}
        (-s_t + f(s_t, t))^\top (s_t- s_t^*) \leq -\mu \lVert s_t-s_t^* \rVert^2.
    \end{equation}
    To prove this, we first define $J_t$ as the integral from the equation above, and, multiplying the previous equation by $(s_t -s^*_t)$, we have
    \begin{align}
        (s_t - s_t^*)^\top (-s_t + f(s_t, t)) &= (s_t - s_t^*)^\top  J_t (s_t - s_t^*)\\
        &= (s_t-s_t^*)^\top \frac{J_t+J_t^\top}{2}(s_t-s_t^*)\\
        & \leq -\mu \lVert s_t-s_t^* \rVert^2 
    \end{align}
    The second equality holds as a quadratic form with a non-symmetric matrix takes the same values as the quadratic form with its symmetric part. The last inequality comes from the assumption on the eigenvalues of the symmetric part of $\partial_s f(s, t)$ being smaller than $1-\mu$.

    Let us now look at the Lyapunov function $V(t) := \frac{1}{2} \lVert s_t - s^*_t \rVert^2$ that measures how far the current estimate $s_t$ is to the equilibrium $s_t^*$. Its temporal derivative is equal to
    \begin{align}
        \dot{V}(t) &= (s_t - s^*_t)^\top \dot{s}_t - (s_t - s^*_t)^\top {\dot{s}^*_t} \\
        &= \tau^{-1}(s_t - s^*_t)^\top (-s_t + f(s_t,t)) - (s_t - s^*_t)^\top \dot{s}^*_t \\
        &\leq -\tau^{-1}\mu \lVert s_t-s^*_t\rVert^2 + \lVert s_t-s^*_t\rVert \lVert \dot{s}^*_t \rVert\\
        &= -2 \tau^{-1} \mu V(t) + \sqrt{2V(t)} \gamma .
    \end{align}

    The inequality is justified by the inequality we proved above and the Cauchy-Schwartz inequality. We can now study the sign of the upper bound on $\dot{V}(t)$ and observe that $\dot{V}(t) < 0$ whenever $V(t) > \frac{\tau^2 \gamma^2}{2\mu^2}$. It then implies that $\underset{t\rightarrow \infty}{\mathrm{limsup}} ~~ V(t) = \frac{\tau^2 \gamma^2}{2\mu^2}$, which is the desired result.

    We now show that the bound is tight by exhibiting an example for which the inequality in the Theorem statement is an equality. For that, we take $t = t$ and follow the dynamics $\dot{s}_t = -s_t/2+t$ ($f(s, t)= s /2 + t$, $\tau=1, \mu=1/2$), we obtain $s^*_t = 2t$, so that $\frac{\gamma \tau}{\mu} = 4$. Alternatively, solving the differential equation gives $s_t = (s_0 + 4)\exp(-t/2) + 2t - 4$ so $\lVert s_t^*- s_t\rVert \rightarrow 4$.
\end{proof}

\paragraph{Remarks.} Let us remark on several things:
\begin{itemize}
    \item The result and the proof are inspired by a result from \cite[Chapter 6]{polyak_introduction_1987} for time-varying optimization, which assumes $f(s, t)=\nabla_s E(s,t)$. We extend this result to general $f$, which requires changing the $\mu$ strong convexity assumption needed in the time-varying optimization setting to the assumption we have on the eigenvalues of the Jacobian in Theorem~\ref{thm:problem_flow}. A discrete-time version of this result can be obtained by adapting the proof of \cite{popkov_gradient_2005}. 
    \item In the main text, we claimed that $\dot{s}^*_t$ both increases with how fast external input varies and depends on the geometry of $\partial_s f_\theta$. This is because the implicit function theorem applied to $s^*_t - f_\theta(s^*_t, t)=0$ gives
    \begin{equation}
        \dot{s}_t^* = - \left (\mathrm{Id} - \frac{\partial f_\theta}{\partial s}(s^*_t, t)\right)^{-1} \frac{\partial f_\theta}{\partial t}(s^*_t, t).
    \end{equation}
    The inverse matrix reflects the geometry of the network and $\partial_t f_\theta$ how fast the inputs to the neurons change.
\end{itemize}

\subsection{Proof of Theorem~\ref{thm:perfect_tracking}}
\begin{thm}
    Let $s$ follow the prospective dynamics (\ref{eqn:prospective_dynamics}). Then, assuming that $\partial_s f_\theta(s_t,t)$ is always invertible and $f_\theta$ is Lipschitz continuous in $t$ on that trajectory, we have 
    \begin{equation*}
        \lVert s_t - f_\theta(s_t, t) \rVert = \lVert s_0 - f_\theta(s_0, 0) \rVert \exp \left (-\frac{t}{\tau} \right )
    \end{equation*}
    and
    \begin{equation*}
        \underset{t \rightarrow \infty}{\mathrm{limsup}} ~ \lVert s_t - s_t^\ast \rVert = 0,
    \end{equation*}
    i.e., after an initial exponential convergence phase, $s_t$ and $s_t^*$ coincide.
\end{thm}
\begin{proof}
The first result follows from rewriting Equation~\ref{eqn:prospective_dynamics} as $\tau \mathrm{d}_t \left [ s_t - f_\theta(s_t, t) \right] = - \left [s_t - f_\theta(s_t, t) \right ]$. It follows that $s_t - f_\theta(s_t, t) = c_0 e^{-\frac{t-t_0}{\tau}}$ , $\lim_{t\rightarrow \infty} \lVert s_t - f_\theta(s_t, t) \rVert = 0$ and  ${\mathrm{limsup}}_{t \rightarrow \infty} ~ \lVert s_t - s_t^* \rVert = 0$.
The assumptions on $f$ ensure that $s_t^*$ is well defined and evolves continuously.
\end{proof}

\section{Physical implementation of prospective neurons}
\label{app:physical_implementation}

\subsection{Adaptive neurons can be prospective}  In the main text, we have claimed that for $u$ and $a$ following the adaptive neuron dynamics
\begin{equation}
    \begin{split}
        \tau \dot{u}_t &= -u_t + \left ( 1 + \frac{\tau}{\tau_a} \right )f_\theta(u_t, t)  - \frac{\tau}{\tau_a} a_t\\
        \tau_a \dot{a}_t &= -a_t + f_\theta(u_t,t),
    \end{split}
\end{equation}
we have
\begin{equation}
    \label{app_eqn:perturbation_adaptive}
    \tau \dot{u}_t = -u_t + f_\theta(u_t, t) + \tau \mathrm{d}_t f_\theta(u_t, t) + \tau \tau_a \mathrm{d}_t^2 f_\theta(u_t, t) + O(\tau \tau_a^2).
\end{equation}
We prove this statement. First, remark that $a_t$ is a low-pass filter version of $f_\theta(u_t, t)$, with a time-constant $\tau_a$ so that
\begin{equation}
    a_t = \frac{1}{\tau_a}\int_{0}^t \exp\left(-\frac{t-t'}{\tau_a}\right) f_\theta(u_{t'},t') \mathrm{d}t'.
\end{equation}
Integration by parts (integrating the exponential, differentiating $f_\theta$) then gives
\begin{align}
    a_t &= \left [ \exp\left(-\frac{t-t'}{\tau_a} \right) f_\theta(u_{t'},t') \right ]_0^t - \int_0^t \exp\left(-\frac{t-t'}{\tau_a} \right) \frac{\mathrm{d}}{\mathrm{d}t} \left [ f_\theta(u_{t'},t') \right ] \mathrm{d}t'\\
    &= f_\theta(u_{t},t) - \exp\left(-\frac{t}{\tau_a}\right)f_\theta(u_{0},0) - \int_0^t \exp\left(-\frac{t-t'}{\tau_a}\right) \frac{\mathrm{d}}{\mathrm{d}t} \left [ f_\theta(u_{t'},t') \right ] \mathrm{d}t'.
\end{align}
Therefore, assuming that $t$ is sufficiently far from the boundary condition $t = 0$, 
\begin{equation}
    \label{app_eqn:adaptive_derivative}
    \frac{f_\theta(u_{t}, t) - a_t}{\tau_a} =\frac{1}{\tau_a} \int^t_0 \exp\left (-\frac{t- t'}{\tau_a}\right ) \frac{\mathrm{d}}{\mathrm{d}t} \left [ f_\theta(u_{t'},t') \right ] \mathrm{d}t',
\end{equation}
holds. We can now use this result to derive Eq.~\ref{app_eqn:perturbation_adaptive}. Let us first remark that $\tau_a^{-1} \left ( f_\theta(u_t, t) - a_t \right )$ converges to $\mathrm{d}_t f_\theta(u_t, t)$ when $\tau_a$ goes to 0. We are interested in the first-order error in $\tau_a$between these two quantities:
\begin{align}
    & \frac{f_\theta(u_t, t) - a_t}{\tau_a} - \frac{\mathrm{d}}{\mathrm{d}t} \left [ f_\theta(u_t, t) \right ] \\ 
    & ~~~ = \frac{1}{\tau_a}\int^t_0 \exp\left (-\frac{t- t'}{\tau_a}\right ) \left ( \frac{\mathrm{d}}{\mathrm{d}t} \left [ f_\theta(u_{t'},t') \right ] - \frac{\mathrm{d}}{\mathrm{d}t} \left [ f_\theta(u_{t},t) \right ] \right )\mathrm{d}t'\\
    & ~~~ = \frac{1}{\tau_a}\int^t_0 \exp\left (-\frac{t- t'}{\tau_a}\right ) \left ( (t'-t) \frac{\mathrm{d}^2}{\mathrm{d}t^2} \left [ f_\theta(u_{t},t) \right ] + O\left ((t'-t)^2\right) \right )\mathrm{d}t'\\
    & ~~~ = \tau_a \frac{\mathrm{d}^2}{\mathrm{d}t^2} \left [ f_\theta(u_{t},t) \right ] + O(\tau_a^2)
\end{align}
Note that in all these calculations, we have assumed $t \gg 0$. The first equation uses the fact that the integral of $\tau_a^{-1} \exp \left (-\tau_a^{-1} (t-t') \right )$ between $0$ and $\infty$ is 1. The second one leverages the Taylor expansion of $\mathrm{d}_t f_\theta(u_t, t)$. For it to be mathematically rigorous, one must assume that the third-order time derivative of $f_\theta(u_t, t)$ is uniformly bounded and make use of the Lagrange formulation of Taylor's remainder. Finally, the last one comes from the standard integral values
\begin{equation}
    \frac{1}{\tau_a} \int_0^\infty (t'- t) \exp \left (- \frac{t-t'}{\tau_a} \right ) \mathrm{d}t' = \tau_a
\end{equation}
and
\begin{equation}
    \frac{1}{\tau_a} \int_0^\infty  (t'- t)^2 \exp \left (- \frac{t-t'}{\tau_a} \right ) \mathrm{d}t' = 2\tau_a^2.
\end{equation}
It follows that the adaptive neuron dynamics become
\begin{equation}
    \tau \dot{u}_t = -u_t + f_\theta(u_t, t) + \tau \mathrm{d}_t f_\theta(u_t, t) + \tau \tau_a \mathrm{d}_t^2 f_\theta(u_t, t) + O(\tau\tau_a^2).
\end{equation}

\subsection{Time constant mismatch}

Here we derive the first order deviation in $(\tau' - \tau)$ of the trajectory $u_t$ when there is a mismatch between the prospective input and the neuron time constants:
\begin{equation}
    \label{eq:recall_diffeq}
    \tau \dot{u}_t = -u_t + f_\theta(u_t, t) + \tau' \frac{\mathrm{d}}{\mathrm{d}t} \left [ f_\theta(u_t, t) \right ].
\end{equation}
We start by making the difference $(\tau' - \tau)$ appear in Eq. \ref{eq:recall_diffeq}:
\begin{equation}
    \tau \dot{u}_t = -u_t + f_\theta(u, t) + \tau \frac{\mathrm{d}}{\mathrm{d}t} \left [ f_\theta(u_t,t) \right ] + (\tau'-\tau) \frac{\mathrm{d}}{\mathrm{d}t} \left [ f_\theta(u_t,t) \right ]. 
\end{equation}
We see that the equation without mismatch appears, and we know by Theorem~2 that $u^{\ast}_t$ is the solution of this equation after a transitory period.
Therefore, we assume that $u_0 = u_0^*$ at initialization, which allows us to expand the mismatched trajectory as:
\begin{equation}
    \label{eq:expansion_ut}
    u_t = u^{\ast}_t + (\tau'-\tau){\rm d}_\tau u_t + O\left ((\tau-\tau')^2 \right ).
\end{equation}
We can perform the same expansion for the other terms appearing in Eq. \ref{eq:recall_diffeq}:
\begin{equation}
    \label{eq:expansion_fut}
    f_\theta(u_t,t) = f_\theta(u^{\ast}_t, t) + (\tau' - \tau) \frac{\partial f_\theta}{\partial u}(u^{\ast}_t, t) \frac{{\rm d} u_t}{{\rm d}\tau} +  O\left ((\tau'-\tau)^2 \right ),
\end{equation}
and
\begin{equation}
    \label{eq:expansion_dtfut}
    \frac{\mathrm{d}}{\mathrm{d}t} \left [ f_\theta(u_t, t) \right ]= \frac{\mathrm{d}}{\mathrm{d}t} \left [ f_\theta(u^{\ast}_t, t) \right ] + (\tau' - \tau) \frac{\partial}{\partial u} \left [ \frac{\mathrm{d}}{\mathrm{d}t} \left [ f_\theta(u^{\ast}_t, t) \right ] \right ] \frac{{\rm d} u_t}{{\rm d}\tau} + O\left ((\tau'-\tau)^2 \right ).
\end{equation}
Then, by plugging Eq. \ref{eq:expansion_ut}, \ref{eq:expansion_fut}, and \ref{eq:expansion_dtfut} into Eq. \ref{eq:recall_diffeq}, and keeping only the terms in $(\tau' - \tau)$ we obtain the following differential equation on ${\rm d}_\tau u_t$:
\begin{equation}
    \label{eq:mismatch_diffeq_general}
    \tau \dot{\frac{{\rm d} u_t}{{\rm d} \tau}} = \left ( -\mathrm{Id} + \frac{\partial}{\partial u} \left  [f_\theta(u^{\ast}_t, t) + \tau \frac{\mathrm{d}}{\mathrm{d}t} \left [ f_\theta(u^{\ast}_t, t) \right ] \right ] \right ) \frac{{\rm d}u_t}{{\rm d}\tau} + \frac{\mathrm{d}}{\mathrm{d}t} \left [ f_\theta(u^{\ast}_t, t) \right ].
\end{equation}
We now proceed to compute the terms appearing in (\ref{eq:mismatch_diffeq_general}) in order to solve it.
First, by definition of equilibrium we have
\begin{equation}
    \frac{\mathrm{d}}{\mathrm{d}t} \left [ f_\theta(u^{\ast}_t,t) \right ] = \dot{u}^{\ast}_t.
\end{equation}
However, replacing $f_\theta(u^{\ast}_t,t)$ by $u^{\ast}_t$ cannot be used when we compute partial derivatives with respect to $u$.
This is because both functions are equal uniformly across time, while having different variations with respect to $u$.
To keep our analysis simple, we now turn to the case $f_\theta(u_t, t) = W\rho(u_t) + W_{\rm in} x_t$. 
We have that the Jacobian at the equilibrium trajectory is equal to
\begin{equation}
    \frac{\partial f_\theta}{\partial u}(u^{\ast}_t, t) = W \odot {\rm diag}(\rho'(u^{\ast}_t)) := W'(t).
\end{equation}
Finally, the Jacobian of the time derivative of $f_\theta(u_t, t)$ is
\begin{align}
    \tau \frac{\partial}{\partial u} &\left [ \frac{\mathrm{d}}{\mathrm{d}t} \left [f_\theta(u^{\ast}_t,t) \right ] \right ] \\&= \tau \frac{\partial}{\partial u} \left [(W_{\rm in} \dot{x}_t + W \rho'(u^{\ast}_t)\odot\dot{u}^{\ast}_t) \right ] \\
    &= \tau W \frac{\partial}{\partial u} \left [ \rho'(u^{\ast}_t)\odot\dot{u}^{\ast}_t \right ] \\
    &= \tau W {\rm diag}(\rho''(u^{\ast}_t)\odot\dot{u}^{\ast}_t) + \tau W'(t) \frac{\partial \dot{u}^{\ast}_t}{\partial u} \\ 
    &= W{\rm diag} \left (\rho''(u^{\ast}_t) \odot \left (-u^{\ast}_t +W_{\rm in} x_t + W \rho(u^{\ast}_t) +\tau \frac{{\rm d}}{{\rm d}t}\left [f_\theta(u^{\ast}_t,t)\right ]\right ) \right ) \label{eq:wsecond}\\
    &+ W'(t) \frac{\partial}{\partial u} \left [-u^{\ast}_t +W_{\rm in} x_t + W \rho(u^{\ast}_t) +\tau \frac{\rm d}{{\rm d} t} \left [ f_\theta(u^{\ast}_t,t) \right ] \right ].
\end{align}
Here we used the following matrix calculus identity
\begin{equation}
    \frac{\partial}{\partial x} \left [ f(x)\odot g(x) \right ] = {\rm diag}(g(x)) \cdot \frac{\partial f}{\partial x}   +  {\rm diag}(f(x)) \cdot  \frac{\partial g}{\partial x}.
\end{equation}
We call $W''(t)$ the term of \eqref{eq:wsecond}.
Then, we continue with
\begin{align}
    \tau \frac{\partial}{\partial u} \left [ \frac{\rm d}{{\rm d}t} \left [f_\theta(u^{\ast}_t,t) \right ] \right ] &= W''(t) \\ &~~~+ W '(t) \frac{\partial}{\partial u} \left [-u^{\ast}_t +W_{\rm in} x_t + W \rho(u^{\ast}_t) +\tau \frac{\rm d}{{\rm d} t} \left [ f_\theta(u^{\ast}_t,t) \right ] \right ] \\
    &= W''(t) - W'(t) + W'(t)^2 + \tau W'(t) \frac{\partial }{\partial u} \left [ \frac{\rm d}{{\rm d}t} \left [f_\theta(u^{\ast}_t,t) \right ] \right ]
\end{align}
so that 
\begin{equation}
    \tau ({\rm Id} - W'(t)) \frac{\partial}{\partial u} \left [ \frac{\rm d}{{\rm d}t} \left [f_\theta(u^{\ast}_t,t) \right ] \right ] = W''(t) - ({\rm Id} - W'(t)) W'(t)
\end{equation}
and
\begin{equation}
     \tau \frac{\partial}{\partial u} \left [ \frac{\rm d}{{\rm d}t} \left [f_\theta(u^{\ast}_t,t) \right ] \right ] = \left ({\rm Id} - W'(t)\right)^{-1}W''(t) - W'(t).
\end{equation}
Here, we assume that the activation function $\rho$ has a negligible second order derivative $\rho'' \approx 0$.
This approximation is actually an equality in cases where the activation function is piece-wise linear, such as the hard sigmoid or the ReLU.
This assumption makes $W''(t)=0$ for all $t$ and we simply have:
\begin{equation}
    \tau \frac{\partial}{\partial u} \left [ \frac{\rm d}{{\rm d}t} \left [f_\theta(u^{\ast}_t,t) \right ] \right ] =  - W'(t).
\end{equation}
We are now in a position to replace all the terms in \eqref{eq:mismatch_diffeq_general}:
\begin{align}
    \tau \dot{\frac{{\rm d}u_t}{{\rm d}\tau}} &= \left(-\mathrm{Id} + W'(t) - W'(t) \right) \frac{{\rm d}u_t}{{\rm d}\tau} + \dot{u}^{\ast}_t \\ 
     \tau \dot{\frac{{\rm d}u_t}{{\rm d}\tau}} &= - \frac{{\rm d}u_t}{{\rm d}\tau} + \dot{u}_t^\ast. 
\end{align}
Solving this differential equation, and after a transitory period ($t \gg 0$), the mismatch is given by:
\begin{equation}
    \frac{{\rm d}u_t}{{\rm d}\tau} = \frac{1}{\tau} \int^t_0 \exp \left(-\frac{t-t'}{\tau} \right)\dot{u}^{\ast}_t ~{\rm d}t'.
\end{equation}
Assuming $\gamma := \max_t \lVert \dot{u}_t^\ast \rVert$, we thus have
\begin{equation}
    \left \lVert \frac{{\rm d}u_t}{{\rm d}\tau} \right \rVert \leq \gamma,
\end{equation}
that is, using Equation~\ref{eq:expansion_ut},
\begin{equation}
    \lVert u_t - u_t^\ast \rVert \leq |\tau' - \tau| \gamma + O\left((\tau - \tau')^2 \right ).
\end{equation}

We emphasize that this is a local result around $\tau' = \tau$ and that the picture may be qualitatively different further away from $\tau$. Let us consider the following example:
\begin{equation}
    \tau \dot u = -u + w u + w \tau' \dot u,
\end{equation}
which is obtained by setting $f(u) = wu$, so that 
\begin{equation}
    (\tau - w \tau') \dot{u} = -(1-w) u
\end{equation}
For $\tau' = \tau$, $u$ converges to $0$ as long as $w < 1$. However, whenever $\tau' > \tau / w$, the dynamics becomes unstable and $u$ exponentially diverges. This highlights that the bound derived above no longer holds in this region. Finally, we note that the case $\tau' = \tau / w$ is degenerate, as only a constant state equal to $0$ satisfies the differential equation (assuming $w \neq 1$). Our analysis does not consider this case.

\section{Additional details on holomorphic equilibrium propagation}
\label{app:hep}

Building on holomorphic equilibrium propagation \citep{laborieux_holomorphic_2022, laborieux_improving_2024}, we demonstrate how prospective dynamics enable neurons to track oscillating feedback signals. In this framework, error terms $\delta$ are computed by introducing finite-size oscillating feedback $\beta_t$ in the complex plane. The key challenge is that oscillations must be slow enough for neurons to equilibrate, introducing a timescale constraint. We now show that prospective dynamics circumvent this limitation, allowing effective tracking regardless of the neurons' intrinsic timescales or the oscillation frequency.

We consider the following neural dynamics:
\begin{equation}
    \tau_u \dot{u}_t = -u_t + W\rho(u_t) + \beta_t (u_t^\mathrm{out} - y_t), ~~\mathrm{with}~~ u_t^\mathrm{in} = x_t.
\end{equation}
Here, $u$ and $\beta$ are complex-valued vectors that evolve over time.
Compared to the one we used as an example in Section~\ref{sec:theory}, the dynamics contains two additional elements: First, the weight matrix includes both feedforward connections that are responsible for processing the input and backward connections that send teaching signals from the output to hidden neurons. Second, the dynamics includes a nudging force $\beta_t (u^\mathrm{out}_t - y_t)$ that pushes output neurons towards the target $y_t$.
The error gradient is multiplied by the oscillating nudging strength $\beta_t = |\beta| \exp(\frac{2 \mathrm{i} \pi t}{\tau_\beta})$. 
Provided that $u_t$ tracks the equilibrium points $u_t^*$ imposed by the feedback $\beta_t$, and that $\beta$-oscillations are much faster than $x$ and $y$, the error is encoded in the first mode of the oscillations:
\begin{equation}
    \delta_t = \frac{1}{\tau_\beta |\beta|}\int_{t-\tau_\beta}^{t} u_{t'} e^{-2 \mathrm{i}\pi t'/ \tau_\beta} \mathrm{d}t'  \,.
\end{equation}
The weight update is the same as Eq.~\ref{eq:weight_update_eq_bp}, where the pre-synaptic term $\rho(u)_i$ is measured as the average of the oscillations induced by $\beta_t$. \citep{laborieux_improving_2024} show that this update closely approximates the gradient.
Whether such complex neurons can be implemented in the brain remains an open question. Alternatively, \cite{anisetti_frequency_2024} have shown that oscillations can exactly compute the gradient in the weak nudging limit. 
Finally, we argue that the separation of timescales needed here is less restrictive than the one assuming infinitely fast neurons; here, only the teaching signal, a small part of the network, has to be fast.

\section{Additional simulation details for the inverted pendulum}
\label{app:additional_details_pendulum}

We use \texttt{gymnax} \citep{lange_gymnax_2022} to simulate the inverted pendulum system through Euler integration ($\Delta t = 1$ms). At each timestep, the agent receives the state $ s_t \in \mathbb{R}^4$ consists of position, velocity, angle, angular velocity ($s_t = [x_t, \dot{x}_t, \theta_t, \dot{\theta}_t]$). An episode terminates when the state vector is outside of the range $-2.4 \leq x \leq 2.4, -12^\circ \leq \theta \leq 12^\circ$.
Typically, for the task, the agent receives a reward of 1 at every time step within the admissible region. Due to our online near-continuous learning regime, we reformulate this as a reward rate 
\begin{equation}
    r(s(t)) = 
    \begin{cases}
        1, & \text{if } -2 \leq x \leq 2, -8.5^\circ \leq \theta \leq 8.5^\circ\\
        0, & \text{otherwise}
    \end{cases}
\end{equation}
Note that prospective neurons take time to react to discontinuity. As a result, it is impossible to learn the right policy solely due to the lack of reward at the last timestep before the end of the episode. Therefore, we use a reward rate that restricts the agent to a smaller rewarded region than the environment's termination condition, allowing for more time steps spent with a reward of $0$.

To train the agent, we adopt the continuous RL formalization, where the value function at time $t$ is defined by
\begin{equation}
    V(s(t)) = \int_t^\infty \exp\left(-\frac{t'-t}{\tau_v}\right) r(s(t')) \mathrm{d}t'
\end{equation}
$\tau_v$ controls the discounting of future rewards. 

\begin{align}
    \dot{V}(t) &= \lim_{\Delta t \rightarrow 0} \frac{1}{\Delta t}\left[\int_{t+\Delta t}^{\infty} \exp\left(-\frac{t'-t-\Delta t}{\tau_v}\right)r(t')\mathrm{d}t' \right. \\
    &~~~~~~~~~~~~~~~~~~~~~~~~~~~~~~~~~~~~~ \left . - \int_{t}^{\infty}\exp\left(-\frac{t'-t}{\tau_v}\right)r(t')\mathrm{d}t'\right]\\
    &= \lim_{\Delta t \rightarrow 0} \frac{1}{\Delta t}\left[- \int_{t}^{t+\Delta t}\exp\left(-\frac{t'-t}{\tau_v}\right)r(t')\mathrm{d}t'  \right. \\
    &~~~~~~~~~~~~~~~~~~~~~~~~~~~~~~~~~~~~~ \left . + \exp\left(\frac{\Delta t}{\tau_v}\right)\int_{t+\Delta t}^{\infty} \exp\left(-\frac{t'-t}{\tau_v}\right)r(t')\mathrm{d}t' \right]\\
    &= \lim_{\Delta t \rightarrow 0} \frac{1}{\Delta t}\left(- \int_{t}^{t+\Delta t}\exp\left(-\frac{t'-t}{\tau_v}\right)r(t')\mathrm{d}t' + \exp\left(\frac{\Delta t}{\tau_v}\right)V(t) \right)\\
    &= -r(t) + \frac{1}{\tau_v}V(t)
\end{align}
Therefore, the continuous TD error is
\begin{equation}
    \label{eqn:critic_cont}
    \delta(t) = r(t) - \frac{1}{\tau_v}V(t) + \dot{V}(t)
\end{equation}
To train the critic, we want to minimize $$E(t) = \frac{1}{2}|\delta(t)|^2$$
Discretizing Eq~\ref{eqn:critic_cont}, we estimate $\dot{V}$ using finite differences
\begin{align}
    \delta(t) &= r(t) - \frac{1}{\tau}V(t) + \frac{V(t) - V(t-\Delta t)}{\Delta t}\\
    &= r(t) + \frac{1}{\Delta t}\left(\left(1 - \frac{\Delta t}{\tau_v}\right)V(t) - V(t-\Delta t)\right)
\end{align}
\begin{equation}
    \hat{\delta}_t =  \delta(t)\Delta t =  r(t)\Delta t + \left(1 - \frac{\Delta t}{\tau_v}\right)V(t) - V(t-\Delta t)
\end{equation}
We thus obtain the appropriate scaling for the reward per timestep $\hat{r}_t = r(t)\Delta t$ and the discount rate $\gamma = 1 - \frac{\Delta t}{\tau_v}$.
We train the agent with the prospective version of the online A2C as detailed in Algorithm \ref{alg:a2c}. Conceptually, we denote the activations (including the network outputs) of the policy and value networks as column vectors ${u}_{\mathrm{p}, t} \in \mathbb{R}^{66}$ ($64$ hidden neurons, $2$ action neurons),  ${u}_{\mathrm{v}, t} \in \mathbb{R}^{65}$ ($64$ hidden neurons, $1$ value neuron) respectively. The feedforward inputs at each step are then given by
\begin{equation}
    f_{\mathrm{p}/\mathrm{v}}\left( {s}_t,  {u}_{\mathrm{p}/ \mathrm{v}, t}\right) = W_{\mathrm{p}/ \mathrm{v}}
    \begin{pmatrix}
         {s}_t\\
        \operatorname{ReLU}\left( {u}_{\mathrm{p}/\mathrm{v}, t}\right)
    \end{pmatrix}
    +  {b}_{\mathrm{p}/\mathrm{v}}
\end{equation}
where $W$ are lower triangular block matrices parameterized appropriately to carry out feedforward computation.  

The loss functions are
\begin{align}
    &\hat{\delta}_t =  r(t)\Delta t + \left(1 - \frac{\Delta t}{\tau_v}\right){u}^{\mathrm{out}}_{\mathrm{v}, t} - {u}_{\mathrm{v}, t-1}^{\mathrm{out}}\\
    &\mathcal{L}_{\mathrm{actor}}\left( {u}_{\mathrm{p}, t}, a\right) = -\hat{\delta_t}\log \left(\left[\operatorname{Softmax}\left( {u}^{\mathrm{out}}_{\mathrm{p}, t}\right)\right]_{a}\right) \\
    &\mathcal{L}_{\mathrm{critic}}\left( {u}_{\mathrm{v}, t}, r\right) = \hat{\delta_t}^2
\end{align}
The optimizer we use is SGD with early stopping (\texttt{critic\_lr}$=1e^{-6}$, \texttt{actor\_lr}$=1e^{-4}$, $\tau_v=1$, $\tau_u=\tau_\delta=1$) with a batch size of 1. We train the agent for $10000$ episodes, and each episode has a maximum time of $T=4$s ($4000$ steps).
\begin{algorithm}[h!]
\caption{Online advantage actor critic (A2C) with prospective neurons}
\label{alg:a2c}
\begin{algorithmic}[1]
\State Initialize  policy and value networks $f_{\mathrm{p}, \mathrm{v}}$ with parameters $ {\theta_0},  {\phi_0}$
\State Initialize environment with appropriate $\Delta t$
\For{each training episode}
\State Initialize state variables $ {u}_{\mathrm{p} / \mathrm{v}, 0}$
\State Initialize gradient variables $ {\delta}_{\mathrm{p} / \mathrm{v}, 0}$ 
\State Reset env and observe the initial state $ {s}_0$
\While{not done \textbf{and} $t<$ \texttt{MAX\_LENGTH}}
    \State $a_t \sim  {u}^{\mathrm{out}}_{\mathrm{p}, t}$ \Comment{Sample action}
    \State $r_t,  {s}_{t+1} \leftarrow \operatorname{env}\left( {s}_t, a_t\right)$ \Comment{Step environment}
    \vspace{0.4cm}
    \State Estimate TD error with prospective dynamics of $ {u}_\mathrm{v}$:
    \State \vspace{-0.8cm}\begin{multline*}
        ~~~~~~~{u}_{\mathrm{v},t+1} \leftarrow  {u}_{\mathrm{v},t} + \frac{\Delta t}{\tau} \left(- {u}_{\mathrm{v},t} + f_\mathrm{v}\left( {s}_{t+1},  {u}_{\mathrm{v},t}\right)\right) \\+ \left(f_\mathrm{v}\left( {s}_{t+1},  {u}_{\mathrm{v},t}\right) - f_\mathrm{v}\left( {s}_t,  {u}_{\mathrm{v},t-1}\right)\right)\end{multline*}
    \vspace{-0.8cm}
    \State $A_t \leftarrow r_t + \left(1-\frac{\Delta t}{\tau_v}\right)  {u}^{\mathrm{out}}_{\mathrm{v},t+1} -  {u}^{\mathrm{out}}_{\mathrm{v},t}$
    \vspace{0.4cm}
    \State Estimate prospective dynamics for $ {u}_\mathrm{p}$:
    \State \vspace{-0.8cm}\begin{multline*}
        ~~~~~~~{u}_{\mathrm{p}, t+1} \leftarrow  {u}_{\mathrm{p},t} + \frac{\Delta t}{\tau} \left(- {u}_{\mathrm{p},t} + f_{\mathrm{p}}\left( {s}_{t+1},  {u}_{\mathrm{p},t}\right)\right) \\+ \left(f_{\mathrm{p}}\left( {s}_{t+1},  {u}_{\mathrm{p},t}\right) - f_{\mathrm{p}}\left( {s}_t,  {u}_{\mathrm{p}, t-1}\right)\right)
        \end{multline*}
        \vspace{-0.4cm}
    \State Backpropagated error dynamics:
    \State \vspace{-0.8cm}\begin{multline*}
        ~~~~~~~{\delta}_\mathrm{p, t+1} \leftarrow  {\delta}_{\mathrm{p},t} + \frac{\Delta t}{\tau} \left(- {\delta}_\mathrm{p, t} + \nabla_ {u}\mathcal{L}_{\mathrm{actor}}\left( {u}_{\mathrm{p}, t}, a_t, A_t\right)\right) \\+ \left(\nabla_ {u}\mathcal{L}_{\mathrm{actor}}\left( {u}_{\mathrm{p}, t}, a_t, A_t\right) - \nabla_ {u}\mathcal{L}_{\mathrm{actor}}\left( {u}_{\mathrm{p},t-1}, a_{t-1}, A_{t-1}\right)\right)
    \end{multline*}
    \vspace{-0.8cm}
    \State \vspace{-0.8cm}\begin{multline*}
        ~~~~~~~{\delta}_{\mathrm{v}, t+1} \leftarrow  {\delta}_{\mathrm{v},t} + \frac{\Delta t}{\tau} \left(- {\delta}_{\mathrm{v}, t} + \nabla_ {u}\mathcal{L}_{\mathrm{critic}}\left(  {u}_{\mathrm{v},t}, r_t\right)\right) \\ + \left(\nabla_ {u}\mathcal{L}_{\mathrm{critic}}\left( {u}_{\mathrm{v},t}, r_t\right) - \nabla_{u}\mathcal{L}_{\mathrm{critic}}\left(  {u}_{\mathrm{v},t-1}, r_{t-1}\right)\right)
    \end{multline*}
    \vspace{-0.4cm}
    \State Update parameters:
    \State $\theta_{t+1} \leftarrow \theta_t - \frac{\Delta t}{\tau_\theta}\left( {\partial_\theta f_\mathrm{p}}^\top {\delta}_{\mathrm{p}, t+1}\right)$
    \State $\phi_{t+1} \leftarrow \phi_t - \frac{\Delta t}{\tau_\phi}\left(  {\partial_\phi f_{\mathrm{v}}}^\top  {\delta}_{\mathrm{v},{t+1}}\right)$
\EndWhile
\EndFor
\end{algorithmic}
\end{algorithm} 
\end{document}